\documentclass[journal]{IEEEtran}

\usepackage{array}
\usepackage[caption=false,font=normalsize,labelfont=sf,textfont=sf]{subfig}
\usepackage{textcomp}
\usepackage{stfloats}
\usepackage{url}
\usepackage{verbatim}
\usepackage[utf8]{inputenc}
\usepackage[T1]{fontenc}
\usepackage{url}
\usepackage{cite}
\usepackage[cmex10]{amsmath}
\usepackage{amssymb,amsfonts,mathtools,bm,bbm}
\usepackage{graphicx}
\usepackage{tikz}
\usepackage{pgfplots}
\pgfplotsset{compat=newest}
\pgfplotsset{plot coordinates/math parser=false}
\usepackage{xcolor}
\usepackage{algorithm}
\usepackage{algorithmic}

\newtheorem{lemma}{Lemma}
\newtheorem{prop}{Proposition}

\newtheorem{rem}{Remark}
\newtheorem{assumption}{Assumption}

\newcommand{\rmd}{\mathrm{d}}

\newcommand{\bbP}{\mathbb{P}}\newcommand{\rmp}{\mathrm{p}}

\newcommand{\bfA}{\mathbf{A}}\newcommand{\bfa}{\mathbf{a}}\newcommand{\sfA}{\mathsf{A}}
\newcommand{\bfB}{\mathbf{B}}
\newcommand{\bfC}{\mathbf{C}}

\newcommand{\bfF}{\mathbf{F}}
\newcommand{\bfG}{\mathbf{G}}
\newcommand{\bfH}{\mathbf{H}}\newcommand{\bfh}{\mathbf{h}}
\newcommand{\bfI}{\mathbf{I}}
\newcommand{\bfJ}{\mathbf{J}}
\newcommand{\bfK}{\mathbf{K}}

\newcommand{\bfM}{\mathbf{M}}

\newcommand{\bfP}{\mathbf{P}}
\newcommand{\bfQ}{\mathbf{Q}}
\newcommand{\bfR}{\mathbf{R}}
\newcommand{\bfS}{\mathbf{S}}\newcommand{\bfs}{\mathbf{s}}
\newcommand{\bfT}{\mathbf{T}}
\newcommand{\bfu}{\mathbf{u}}

\newcommand{\bfW}{\mathbf{W}}\newcommand{\bfw}{\mathbf{w}}
\newcommand{\bfX}{\mathbf{X}}\newcommand{\bfx}{\mathbf{x}}
\newcommand{\bfY}{\mathbf{Y}}\newcommand{\bfy}{\mathbf{y}}
\newcommand{\bfZ}{\mathbf{Z}}\newcommand{\bfz}{\mathbf{z}}\newcommand{\sfZ}{\mathsf{Z}}

\newcommand{\cD}{\mathcal{D}}

\newcommand{\cF}{\mathcal{F}}

\newcommand{\cK}{\mathcal{K}}

\newcommand{\cS}{\mathcal{S}}

\newcommand{\cU}{\mathcal{U}}

\newcommand{\C}{\mathbb{C}}
\newcommand{\E}{\mathbb{E}}
\newcommand{\herm}{\dagger}
\newcommand{\tr}[1]{{\rm tr}\!\left(#1\right)}

\newcommand{\statementend}{\hfill$\square$}
\newcommand{\definitionend}{\hfill$\lozenge$}

\newlength\figW
\newlength\figH

\newlength\figurewidth
\newlength\figureheight

\begin{document}

\title{Fundamental Limits of Dynamic Tracking and Communication: \\
A Stationary Policy Formulation}

\author{%
Mohammadreza~Bakhshizadeh~Mohajer,
Luca~Barletta,
Alex~Dytso,
and~Giuseppe~Caire%
\thanks{A preliminary version of this work has been accepted for presentation
at the 2026 IEEE Information Theory Workshop (ITW), Tempe, AZ, USA,
Nov.~10--13, 2026.}%
\thanks{Mohammadreza Bakhshizadeh Mohajer and Luca Barletta are with the
Dipartimento di Elettronica, Informazione e Bioingegneria,
Politecnico di Milano, Milano, Italy
(e-mail: \{mohammadreza.bakhshizadeh,luca.barletta\}@polimi.it).}%
\thanks{Alex Dytso is with Qualcomm  Inc.,
Bridgewater, NJ 08807 USA (e-mail: odytso2@gmail.com).}%
\thanks{Giuseppe Caire is with the Department of Electrical Engineering
and Computer Science, Technical University of Berlin, Berlin, Germany
(e-mail: caire@tu-berlin.de).}%
}

\maketitle

\begin{abstract}
We study a dynamic integrated sensing-and-communication problem in which a
multi-antenna transmitter communicates with a receiver while simultaneously
tracking a moving target governed by a stable Gauss--Markov model. We formulate
the sensing--communication trade-off as a rate--distortion problem, where the
communication rate is defined through the Verd\'u--Han information spectrum and
the sensing distortion is measured by the long-run minimum mean square tracking
error in angle and distance. We construct a tractable outer region by
lower bounding this distortion through the posterior Cram\'er--Rao bound (PCRB). The control problem is partially observed because
the target state is not directly available to the transmitter. We therefore
introduce a belief-augmented information state and show, under the stated
regularity assumptions, that stationary randomized Markov policies suffice for
the covariance-based scalarized problem used in the outer bound. 
We then derive, under the stated large-block conditions, a Gaussian
covariance-control outer bound 
by showing that, in the large-block regime, the sensing belief
transition becomes asymptotically covariance sufficient; consequently,
Gaussian signaling with the same stationary covariance policy preserves the
stationary PCRB asymptotically while maximizing the communication reward for
each transmit covariance. Numerical experiments compare an extended Kalman filter-based 
achievable frontier with a PCRB-based Gaussian covariance-control outer frontier.
\end{abstract}

\begin{IEEEkeywords}
Integrated sensing and communication (ISAC), dynamic tracking, posterior
Cram\'er--Rao bound, information spectrum, Markov decision processes.
\end{IEEEkeywords}

\section{Introduction}

\IEEEPARstart{T}{he} integration of sensing capabilities into communication
networks is expected to unlock a broad range of applications, from autonomous
coordination to environment-aware wireless systems. A central challenge in this
setting is that the waveform design that is desirable for communication is, in
general, not the one that is desirable for sensing, and vice versa. This tension
becomes even more pronounced in dynamic scenarios, where the sensing task is not
just to estimate a static parameter but to track a moving target over time.

In this paper, we study this dynamic problem from an information-theoretic
viewpoint. We consider a MIMO transmitter that communicates with a dedicated
receiver while simultaneously tracking a mobile target whose state evolves
according to a stable Gauss--Markov model. The communication performance is
measured through an information rate, whereas the sensing performance is measured
through a long-run minimum mean square tracking distortion in angle and distance. 
We use the posterior Cram\'er--Rao
bound (PCRB) to construct a tractable outer characterization of this distortion.

The main point of this paper is that the dynamic problem admits a stationary
information-state outer-bound formulation. Since the target state is not directly observed
at the transmitter, the physical state alone is not a valid Markov state for
control. We therefore introduce a belief-augmented state that summarizes the
sensing history relevant for prediction and tracking. Under the stated regularity 
assumptions, we show that the covariance-based
scalarized problem used in the outer bound admits an optimal stationary
randomized Markov policy on this information state.

Prior work on integrated sensing and communication (ISAC) has studied the reuse
of spectrum, hardware, and spatial degrees of freedom for simultaneous
communication and sensing. Surveys and overview papers cover architectures,
waveform and beamforming design, resource allocation, and fundamental-limit
perspectives for ISAC and related dual-functional radar--communication systems
\cite{eldar_survey,jrc_survey,jrc_sp_survey,IT_survey}. From an
information-theoretic viewpoint, recent ISAC formulations have characterized
capacity--distortion or rate--distortion trade-offs for mostly static or
memoryless sensing models
\cite{jssc_kobayashi,itjsc_ahmadi,fund_isac_caire}. These works provide the
basis for treating sensing performance as a distortion constraint, but they do
not address long-run causal tracking of a dynamically evolving target.

For multi-antenna Gaussian ISAC, a closely related line of work studies
CRB--rate trade-offs through transmit covariance optimization. These works
optimize the covariance matrix to maximize the MIMO Gaussian rate subject to CRB
constraints, thereby tracing Pareto boundaries for point and extended target
estimation problems~\cite{mimo_crb_rate,mimo_crb_rate_extended}. More recently,
for the non-tracking setting in which the target parameter remains fixed over a
single large block, \cite{yilmaz2026joint} established the
optimality of Gaussian signaling and reduced the asymptotic
communication--estimation Pareto-boundary characterization from an optimization
over input distributions to an optimization over transmit covariance matrices.
These formulations are closely related to ours because both the Gaussian communication rate and the
Fisher-information sensing metric depend on the transmit covariance matrix.
However, they are primarily one-shot or blockwise formulations, whereas the
problem considered here involves causal tracking of a dynamically evolving
target.

Dynamic tracking brings in Bayesian filtering and long-run control. The PCRB
provides a recursive lower bound for discrete-time nonlinear filtering
\cite{Tichavsky1998}, and recent ISAC tracking works have used extended Kalman filter (EKF) or PCRB-type criteria
for trajectory and resource-allocation design~\cite{uav_no_rate,open_loop}. In
particular,~\cite{open_loop} studies a Markov-state ISAC model, identifies
Bayesian filtering as the optimal sensing strategy, and characterizes the
communication--sensing trade-off for open-loop encoding. Our work differs in the
control structure and objective: we optimize over causal transmission policies
that adapt to the accumulated sensing information. 
This leads to a belief-state controlled Markov formulation and a
stationary-policy characterization of the covariance-based problem entering
the outer bound on the dynamic information-rate--tracking-distortion
trade-off.
Since the optimization in this work is over transmission policies, which may in
general be non-stationary, the resulting communication channel need not be
stationary or information stable a priori. For this reason, the usual
single-letter mutual information is not the appropriate performance metric. We
instead use the information-spectrum framework and measure communication through
the Verd\'u--Han inf-information rate. The general capacity formula for
single-user channels without feedback was established in~\cite{verdu_general},
and extensions to channels with feedback were developed in~\cite{GenCapFeedback}.
A general introduction to information spectrum can be found
in~\cite{han_information_spectrum}.

The paper is organized around the following contributions.
\begin{itemize}

    \item In Sections~\ref{sec:system_model}--\ref{sec:problem_formulation}, we
    formulate a dynamic ISAC rate--distortion problem in which a MIMO transmitter
    communicates with a receiver while tracking a target governed by a stable
    Gauss--Markov model. Communication is measured through the Verd\'u--Han
    inf-information rate, while tracking distortion is defined by the long-run 
    minimum mean square error in angle and distance and lower bounded through the PCRB.

    \item In Section~\ref{sec:stationary_policies}, we introduce a
    belief-augmented information state for the partially observed tracking
    problem and prove, under the stated regularity assumptions, that the
    covariance-based scalarized problem used in the outer bound admits an optimal 
    stationary randomized Markov policy.
    
    \item In Section~\ref{sec:gaussian_outer_bound}, we establish a
    large-block Gaussian covariance reduction. We show that, as the empirical
    transmit covariance concentrates around the selected covariance, the sensing
    belief transition becomes asymptotically covariance sufficient. Consequently,
    for a fixed stationary covariance policy, conditionally Gaussian signaling
    preserves the stationary PCRB asymptotically while achieving asymptotically no
    smaller communication rate than an arbitrary admissible signaling law with the
    same covariance policy.

    \item In Section~\ref{sec:simulations}, we evaluate the resulting 
    trade-off by comparing an EKF-based achievable frontier with a 
    PCRB-based Gaussian covariance-control outer frontier.
\end{itemize}

Relative to the preliminary conference version~\cite{mohajer_itw2026},
the present paper provides the complete technical development and proofs
of the stationary-policy and large-block Gaussian reduction results,
including the finite-block Gram-sufficiency analysis.

Detailed proofs of the stationary-policy and large-block Gaussian results are
provided in the appendices.

\subsection*{Notation}
Deterministic scalars are denoted by italic lowercase letters. Vectors are
denoted by bold lowercase letters, and matrices by bold uppercase letters; calligraphic letters denote sets. For a matrix $\bfA$, the notation $\bfA^{T}$, $\bfA^{\dagger}$, $\bfA^{-1}$, $\det(\bfA)$, $\tr{\bfA}$, and $\|\bfA \|_F$ denote its transpose, Hermitian transpose, inverse, determinant, trace, and Frobenius norm respectively. The identity matrix of appropriate dimension is denoted by $\bfI$, and $\bf0$ denotes a zero vector or matrix. The semidefinite ordering is denoted by $\succeq$. The expectation and probability operators are denoted by $\E [\cdot]$ and $\bbP(\cdot)$, respectively. For random variables $\bfX$ and $\bfY$, $P_{\bfX}$ denotes the distribution of $\bfX$, and $I(\bfX;\bfY|\cF)$ denotes conditional mutual information given $\cF$. The sigma-algebra generated by random variables is denoted by $\sigma(\cdot)$.

\section{System Model}
\label{sec:system_model}

Modern wireless communication systems operate over bandwidths $W$ of tens of
MHz, such that a time domain communication symbol has duration of order $1/W$,
i.e., a fraction of a microsecond. In contrast, dynamic targets moving at a few
km/h change their position and velocity on a time scale of order seconds, which
is at least seven orders of magnitude larger. Consequently, one channel use at
the target dynamics time scale corresponds to $L\gg 1$ channel uses at the
signaling time scale.

This separation of time scales motivates a two-time scale, state dependent
channel model indexed by the tracking interval $k$. We consider an ISAC
transmitter equipped with an $N_t$-antenna array, whose input letters are
$N_t\times L$ space-time blocks
\begin{equation}
    \bfX_k
    =
    [\bfx_{k,1},\ldots,\bfx_{k,L}]
    \in \C^{N_t\times L},
    \label{eq:block_signal}
\end{equation}
where $\bfx_{k,\ell}\in\C^{N_t}$ denotes the transmitted vector at signaling
channel use $\ell$ within tracking interval $k$ and is used jointly for
communication and sensing. The target state is assumed to remain constant
within each tracking interval and evolves from one tracking interval to the
next.

\subsection{Target State Evolution Model}

The target's kinematic state at time $k$ is described by
\begin{equation}
    \bfs_k=[\theta_k,\dot\theta_k,d_k,\dot d_k]^T,
    \label{eq:def_sk}
\end{equation}
where $\theta_k$ is the target angle, $\dot\theta_k$ is angular velocity, $d_k$
is distance, and $\dot d_k$ is radial velocity. We assume the linear stochastic
state-space model
\begin{equation}
    \bfs_k=\bfF\bfs_{k-1}+\bfw_k,
    \label{eq:state_evolution}
\end{equation}
where $\bfw_k\sim \mathcal N(\mathbf 0,\bfQ_w)$. The matrix $\bfF$ is assumed
Schur stable:
\begin{equation}
    \rho(\bfF)<1,
\end{equation}
where $\rho(\bfF)$ denotes the spectral radius of $\bfF$. Hence, the target
state process admits a unique stationary distribution. We initialize
$\bfs_0$ according to the corresponding stationary Gaussian law of
\eqref{eq:state_evolution}.

\subsection{Communication Channel}

At channel use $\ell$ within tracking block $k$, the signal received by the
communication user is
\begin{equation}
    \bfy_{c,k,\ell} = \bfH_{c,k}\bfx_{k,\ell} + \bfz_{c,k,\ell}, \qquad \ell=1,\ldots,L,
    \label{eq:comm_channel}
\end{equation}
where $\bfH_{c,k}\in \C^{N_{r,c}\times N_t}$ is the communication channel
matrix and $\bfz_{c,k,\ell}\sim\mathcal{CN}(\mathbf 0,\sigma_c^2\bfI)$ is the
communication noise. We assume that
$\{\bfH_{c,k}\}_{k\ge 1}$ is i.i.d. across tracking blocks, independent of the
target and noise processes, and that $\bfH_{c,k}$ is available at both the
transmitter and the communication receiver during block $k$.
We denote the communication-output block by
\begin{equation}
    \bfY_{c,k}
    =
    [\bfy_{c,k,1},\ldots,\bfy_{c,k,L}].
\end{equation}

\subsection{Sensing Channel}

At channel use $\ell$ within tracking block $k$, the sensing observation is
\begin{align}
    \bfy_{s,k,\ell}
    &=
    \bfh(\bfs_k,\bfx_{k,\ell})+\bfz_{s,k,\ell} \\
    &=
    \alpha(d_k)\,
    \bfa_{\mathrm{rx}}(\theta_k)
    \bfa_{\mathrm{tx}}^{\herm}(\theta_k)\bfx_{k,\ell}
    +
    \bfz_{s,k,\ell},
    \qquad
    \ell=1,\ldots,L,
    \label{eq:sensing_model_conf}
\end{align}
where
\begin{equation}
\bfa_{\mathrm{tx}}(\theta_k)\in\mathbb{C}^{N_t},
\qquad
\bfa_{\mathrm{rx}}(\theta_k)\in\mathbb{C}^{N_{r,s}}
\end{equation}
denote the transmit and receive steering vectors, respectively. Here, $N_{r,s}$ is the number of sensing receive
antennas. Moreover, $\bfz_{s,k,\ell}\sim \mathcal{CN}(\mathbf 0,\sigma_s^2\bfI)$, where $\sigma_s^2$ is the per channel use sensing-noise variance (independent
of the block length $L$), and
\begin{equation}
    \alpha(d_k)=\beta e^{-j2\pi d_k/\lambda_c}/d_k^2.
\end{equation}
Here $\beta$ is the radar cross-section of the moving target and $\lambda_c$ is the carrier
wavelength. We denote the sensing-output block by
\begin{equation}
    \bfY_{s,k}
    =
    [\bfy_{s,k,1},\ldots,\bfy_{s,k,L}].
\end{equation}
Thus, $\bfX_k$, $\bfY_{c,k}$, and $\bfY_{s,k}$ denote block-level quantities,
whereas $\bfx_{k,\ell}$, $\bfy_{c,k,\ell}$, and $\bfy_{s,k,\ell}$ denote the
corresponding channel-use-level vectors. Communication rate is normalized per
channel use, while sensing and target evolution are indexed at the tracking-block
time scale.

\section{Problem Formulation and Performance Metrics}
\label{sec:problem_formulation}

\subsection{Transmission Policy}

For each $k\geq 1$, let
\begin{equation}
    \cF_{k-1}
    =
    \sigma\!\left(
    \bfY_{s,1},\ldots,\bfY_{s,k-1},
    \bfX_1,\ldots,\bfX_{k-1},
    \bfH_{c,k}
    \right)
\end{equation}
denote the information available at the transmitter before selecting the
transmitted block $\bfX_k$. A causal transmission policy is a sequence of
conditional block-input laws
\begin{equation}
    \pi
    =
    \left\{
    P_{\bfX_k|\cF_{k-1}}
    \right\}_{k\geq 1}.
\end{equation}
Every admissible policy satisfies the per-channel-use power constraint
\begin{equation}
    \frac{1}{L}
    \E\!\left[
    \|\bfX_k\|_F^2
    \,\middle|\,
    \cF_{k-1}
    \right]
    \leq P
    \qquad \text{a.s.}
    \label{eq:power_constraint_conf}
\end{equation}
Let $\Pi_{\rm causal}$ denote the class of admissible causal transmission
policies satisfying~\eqref{eq:power_constraint_conf}.
The corresponding conditional per-channel-use covariance is
\begin{equation}
    \bfQ_{\bfx,k}
    \triangleq
    \frac{1}{L}
    \E\!\left[
    \bfX_k\bfX_k^{\herm}
    \,\middle|\,
    \cF_{k-1}
    \right],
    \label{eq:conditional_covariance}
\end{equation}
which satisfies $\bfQ_{\bfx,k}\succeq\mathbf 0$ and
$\tr{\bfQ_{\bfx,k}}\leq P$ almost surely. We therefore define the covariance
action set
\begin{equation}
    \cU
    \triangleq
    \left\{
    \bfQ\succeq\mathbf 0:
    \tr{\bfQ}\leq P
    \right\}.
    \label{eq:action_set_conf}
\end{equation}
Under full power transmission, the corresponding per channel use transmit
SNRs are $P/\sigma_c^2$ and $P/\sigma_s^2$ for the communication and sensing
channels, respectively.
Since $\bfH_{c,k}$ is contained in $\cF_{k-1}$, the selected block-input law,
and hence $\bfQ_{\bfx,k}$, may depend on the current communication channel.

For a realized transmitted block, define the empirical covariance
\begin{equation}
    \widehat{\bfQ}_{k,L}
    \triangleq
    \frac{1}{L}\bfX_k\bfX_k^{\herm}.
    \label{eq:empirical_block_covariance}
\end{equation}
At finite $L$, the realized empirical covariance
$\widehat{\bfQ}_{k,L}$ need not coincide with the conditional covariance
$\bfQ_{\bfx,k}$ in~\eqref{eq:conditional_covariance}. Their relation in the
large-block regime is introduced in Section~\ref{sec:gaussian_outer_bound}.

\subsection{Communication Rate: Verd\'u--Han Inf-Information Rate}

Due to the adaptive nature of the transmission policy, the input--output
process induced by $\pi$ need not be stationary or information stable. 
Therefore, rather than using a single-letter mutual information
expression, we measure communication using the Verd\'u--Han inf-information rate
\cite{verdu_general,han_information_spectrum}.

For an admissible policy $\pi$, let
$\bfX_{1:N}=(\bfX_1,\ldots,\bfX_N)$,
$\bfY_{c,1:N}=(\bfY_{c,1},\ldots,\bfY_{c,N})$, and
$\bfH_{c,1:N}=(\bfH_{c,1},\ldots,\bfH_{c,N})$.
The corresponding conditional information density is
\begin{align}
    i(\bfX_{1:N};\bfY_{c,1:N}&\mid\bfH_{c,1:N})
    \triangleq \nonumber\\
    &\log_2
    \frac{
    \rmd P_{\bfX_{1:N}\bfY_{c,1:N}\mid\bfH_{c,1:N}}
    }{
    \rmd\left(
    P_{\bfX_{1:N}\mid\bfH_{c,1:N}}
    \times
    P_{\bfY_{c,1:N}\mid\bfH_{c,1:N}}
    \right)
    } .
    \label{eq:conditional_information_density_conf}
\end{align}

The inf-information rate under $\pi$ is\footnote{
For a sequence of real-valued random variables $\{Z_N\}$, the limit inferior in
probability is defined as
$
\rmp\text{-}\liminf_{N\to\infty} Z_N
=
\sup\left\{a\in\mathbb{R}:\lim_{N\to\infty}\bbP(Z_N<a)=0\right\}.
$
}
\begin{equation}
    \underline{R}(\pi)
    =
    \rmp\text{-}\liminf_{N\to\infty}
    \frac{1}{NL}
    i(\bfX_{1:N};\bfY_{c,1:N}\mid\bfH_{c,1:N}).
    \label{eq:inf_information_rate_conf}
\end{equation}

\subsection{Belief State}
The transmitter does not directly observe the target state. For each $k\geq1$,
define the sensing history through block $k$ as
\begin{equation}
    \cS_k
    \triangleq
    \sigma\!\left(
        \bfY_{s,1},\ldots,\bfY_{s,k},
        \bfX_1,\ldots,\bfX_k
    \right),
\end{equation}
with $\cS_0$ denoting the trivial initial sigma-algebra. The predictive belief
before block $k$ is
\begin{equation}
    \Pi_k
    \triangleq
    P_{\bfs_k|\cS_{k-1}},
    \label{eq:belief_state_def}
\end{equation}
whereas the full information $\cF_{k-1}$ available for selecting $\bfX_k$
also contains the current communication channel $\bfH_{c,k}$.

By the independence assumptions of Section~\ref{sec:system_model},
$\bfH_{c,k}$ does not modify the target belief:
\begin{equation}
    P_{\bfs_k|\cS_{k-1},\bfH_{c,k}}
    =
    P_{\bfs_k|\cS_{k-1}}
    =
    \Pi_k .
    \label{eq:belief_channel_independence}
\end{equation}
Thus, $\Pi_k$ summarizes the sensing information relevant to the target,
while $\bfH_{c,k}$ remains separately available for transmission control.

\subsection{Tracking Distortion and PCRB Lower Bound}

We measure tracking performance through the long-run minimum mean-square error
in angle and distance. To obtain a tractable outer bound to the capacity region, we lower-bound
this distortion using the posterior Cram\'er--Rao bound (PCRB)
of~\cite{Tichavsky1998}. We first make explicit the dependence of the sensing
data information on the transmitted covariance.

For compactness, define the state dependent sensing matrix
\begin{equation}
    \bfH_s(\bfs)
    \triangleq
    \alpha(d)\,
    \bfa_{\mathrm{rx}}(\theta)
    \bfa_{\mathrm{tx}}^{\herm}(\theta),
\end{equation}
so that
$\bfh(\bfs,\bfx)=\bfH_s(\bfs)\bfx$
in~\eqref{eq:sensing_model_conf}.
For $i=1,\ldots,4$, let
\begin{equation}
    \bfG_i(\bfs)
    \triangleq
    \frac{\partial \bfH_s(\bfs)}{\partial s_i}.
\end{equation}
For fixed state $\bfs_k$ and transmitted block $\bfX_k$, the conditional
data Fisher information matrix is
\begin{equation}
    \bfJ_{D,k}^{\mathrm{cond}}(\bfs_k,\bfX_k)
    \triangleq
    \E\!\left[
    -\nabla_{\bfs_k}^{2}
    \log p(\bfY_{s,k}|\bfs_k,\bfX_k)
    \,\middle|\,
    \bfs_k,\bfX_k
    \right].
    \label{eq:conditional_data_fim}
\end{equation}
For the complex Gaussian sensing model, its $(i,j)$th entry is
\begin{align}
    \left[
    \bfJ_{D,k}^{\mathrm{cond}}(\bfs_k,\bfX_k)
    \right]_{i,j}
    &=
    \frac{2}{\sigma_s^2}
    \Re\!\left\{
    \sum_{\ell=1}^{L}
    \bfx_{k,\ell}^{\herm}
    \bfG_i^{\herm}(\bfs_k)
    \bfG_j(\bfs_k)
    \bfx_{k,\ell}
    \right\}
    \notag\\
    &=
    \frac{2L}{\sigma_s^2}
    \Re\!\left\{
    \tr{
    \bfG_i^{\herm}(\bfs_k)
    \bfG_j(\bfs_k)
    \widehat{\bfQ}_{k,L}
    }
    \right\}.
    \label{eq:conditional_data_fim_covariance}
\end{align}
Hence, for a fixed state, the block data Fisher information depends on the realized
waveform through its empirical covariance. To average this contribution over
the randomized block-input law, note that under an admissible causal policy,
$\bfX_k$ is conditionally independent of the unobserved state $\bfs_k$ given
$\cF_{k-1}$, i.e.,
\begin{equation}
    \bfX_k \perp \bfs_k | \cF_{k-1}.
    \label{eq:input_state_conditional_independence}
\end{equation}
Using $\widehat{\bfQ}_{k,L}=L^{-1}\bfX_k\bfX_k^{\herm}$,
and noting that $\bfG_i(\bfs_k)$ and $\bfG_j(\bfs_k)$ are fixed when
conditioning on $\bfs_k$, we obtain
\begin{align}
    &\E\!\left[
        \Re\!\left\{
        \tr{
        \bfG_i^{\herm}(\bfs_k)
        \bfG_j(\bfs_k)
        \widehat{\bfQ}_{k,L}
        }
        \right\}
        \,\middle|\,
        \cF_{k-1},\bfs_k
    \right]
    \nonumber\\
    &\qquad=
    \Re\!\left\{
        \tr{
        \bfG_i^{\herm}(\bfs_k)
        \bfG_j(\bfs_k)
        \E\!\left[
            \widehat{\bfQ}_{k,L}
            \,\middle|\,
            \cF_{k-1},\bfs_k
        \right]
        }
    \right\}
    \nonumber\\
    &\qquad=
    \Re\!\left\{
        \tr{
        \bfG_i^{\herm}(\bfs_k)
        \bfG_j(\bfs_k)
        \bfQ_{\bfx,k}
        }
    \right\},
    \label{eq:conditional_fim_waveform_average}
\end{align}
where the last equality follows from
\eqref{eq:input_state_conditional_independence},
\eqref{eq:conditional_covariance}, and
\eqref{eq:empirical_block_covariance}. Since
$P_{\bfs_k|\cF_{k-1}}=\Pi_k$ by
\eqref{eq:belief_channel_independence}, the remaining average over the
predictive state distribution yields the belief-averaged data-information
matrix $\bfJ_{D,k}(\Pi_k,\bfQ_{\bfx,k})$, whose entries are
\begin{align}
    &\left[
    \bfJ_{D,k}(\Pi_k,\bfQ_{\bfx,k})
    \right]_{i,j}
    =\nonumber\\
    &\qquad\frac{2L}{\sigma_s^2}
    \E_{\bfs_k\sim\Pi_k}
    \!\left[
    \Re\!\left\{
    \tr{
    \bfG_i^{\herm}(\bfs_k)
    \bfG_j(\bfs_k)
    \bfQ_{\bfx,k}
    }
    \right\}
    \right].
    \label{eq:data_fim_covariance_form}
\end{align}
Thus, conditional on the predictive belief, the expected data-information
contribution depends on the selected block-input law through its covariance.
This covariance dependence of the Fisher information should not be confused
with covariance sufficiency of the complete sensing experiment: at finite
$L$, the belief transition generally depends on the realized transmitted block,
a distinction used in Section~\ref{sec:gaussian_outer_bound}.

For an admissible causal policy $\pi$, let $\bbP^\pi$ denote the induced joint
law and define the policy-averaged data-information matrix
\begin{equation}
    \bfM_k^\pi
    \triangleq
    \E^\pi\!\left[
        \bfJ_{D,k}(\Pi_k,\bfQ_{\bfx,k})
    \right].
    \label{eq:policy_averaged_data_information}
\end{equation}
The expectation is joint over the policy-induced distribution of
$(\Pi_k,\bfQ_{\bfx,k})$; no independence between the predictive belief and
the adaptive covariance action is assumed.

Let $\bfJ_0$ denote the Fisher information matrix of the stationary Gaussian
initial-state distribution. For the linear Gaussian dynamics in
\eqref{eq:state_evolution}, the posterior-information recursion~\cite{Tichavsky1998}
reduces to
\begin{equation}
    \bfJ_k^\pi
    =
    \left(
    \bfQ_w
    +
    \bfF(\bfJ_{k-1}^\pi)^{-1}\bfF^T
    \right)^{-1}
    +
    \bfM_k^\pi .
    \label{eq:PCRB_recursion}
\end{equation}
For a fixed policy $\pi$, the sequence
$\{\bfJ_k^\pi\}_{k\geq0}$ is deterministic because all randomness due to the
target trajectory, sensing observations, transmitted waveforms, and policy
randomization has already been averaged under $\bbP^\pi$.

For a fixed policy $\pi$, define the Bayesian MMSE estimate of the target state after block $k$ as
\begin{equation}
    \widehat{\bfs}_k^{\rm MMSE}
    \triangleq
    \E^\pi\!\left[
        \bfs_k
        \,\middle|\,
        \cS_k
    \right].
    \label{eq:mmse_estimator_conf}
\end{equation}
The posterior Cram\'er--Rao bound implies
\begin{equation}
    \E^\pi\!\left[
        (\widehat{\bfs}_k^{\rm MMSE}-\bfs_k)
        (\widehat{\bfs}_k^{\rm MMSE}-\bfs_k)^T
    \right]
    \succeq
    (\bfJ_k^\pi)^{-1}.
    \label{eq:true_pcrb_mse_bound}
\end{equation}

Let
\begin{equation}
    \bfS
    =
    {\rm diag}([1,0,1,0]),
\end{equation}
which selects the angle and distance components of $\bfs_k$. The corresponding
long-run MMSE tracking distortion is
\begin{equation}
    \overline{D}_{\rm MMSE}(\pi)
    \triangleq
    \limsup_{N\to\infty}
    \frac{1}{N}
    \sum_{k=1}^{N}
    \E^\pi\!\left[
        \left\|
            \bfS
            (\widehat{\bfs}_k^{\rm MMSE}-\bfs_k)
        \right\|_2^2
    \right].
    \label{eq:long_run_mmse_distortion}
\end{equation}
The associated long-run PCRB-based distortion is
\begin{equation}
    \overline{D}(\pi)
    \triangleq
    \limsup_{N\to\infty}
    \frac{1}{N}
    \sum_{k=1}^{N}
    \tr{
        \bfS
        (\bfJ_k^\pi)^{-1}
        \bfS^T
    }.
    \label{eq:tracking_distortion_conf}
\end{equation}
By~\eqref{eq:true_pcrb_mse_bound},
\begin{equation}
    \overline{D}_{\rm MMSE}(\pi)
    \geq
    \overline{D}(\pi).
    \label{eq:long_run_pcrb_bound}
\end{equation}

\subsection{Rate--Distortion Region and Scalarized Objective}

A rate--distortion pair $(R,D)$ is achievable if there exists
$\pi\in\Pi_{\rm causal}$ such that
\begin{equation}
    \underline{R}(\pi)\geq R,
    \qquad
    \overline{D}_{\rm MMSE}(\pi)\leq D.
    \label{eq:achievable_rd_pair}
\end{equation}
The rate--distortion region $\mathcal C$ is the closure of all achievable pairs.
We consider a PCRB-based outer region
\begin{equation}
    \overline{\mathcal C}
    \triangleq
    \operatorname{cl}
    \left(
    \bigcup_{\pi\in\Pi_{\rm causal}}
    \left\{
        (R,D):
        \underline R(\pi)\geq R,\;
        \overline{D}(\pi)\leq D
    \right\}
    \right).
    \label{eq:pcrb_outer_region}
\end{equation}
By~\eqref{eq:long_run_pcrb_bound}, the policy specific PCRB relaxation
contains the corresponding achievable pairs. Taking the union over
$\Pi_{\rm causal}$ and then the closure therefore gives
\begin{equation}
    \mathcal C
    \subseteq
    \overline{\mathcal C}.
    \label{eq:rd_outer_inclusion}
\end{equation}

To characterize supporting boundary points of $\overline{\mathcal C}$, we consider the scalarized objective
\begin{equation}
    J_{\lambda}(\pi)
    \triangleq
    \underline{R}(\pi)-\lambda \overline{D}(\pi),
    \qquad \lambda\ge 0,
    \label{eq:J_lambda_conf}
\end{equation}
with optimal value
\begin{equation}
    J_{\lambda}^{\star}
    \triangleq
    \sup_{\pi\in\Pi_{\rm causal}} J_{\lambda}(\pi).
    \label{eq:J_star_conf}
\end{equation}
In the following sections, we derive a covariance-based outer relaxation of this problem and establish the sufficiency of stationary policies 
for that outer problem.

\section{Stationary Policy Formulation}
\label{sec:stationary_policies}

We next derive a stationary-policy formulation of the covariance-based
problem used in the outer relaxation of the rate--distortion problem in
Section~\ref{sec:problem_formulation}. 

\subsection{Information State and Communication Upper Reward}
\label{sec:augmented_state_actions}

Using the predictive belief $\Pi_k$ defined in~\eqref{eq:belief_state_def},
we define the pre-decision information state as
\begin{equation}
    \bfZ_k
    \triangleq
    \left(\Pi_k,\bfH_{c,k}\right).
    \label{eq:augmented_belief_state}
\end{equation}
The belief $\Pi_k$ summarizes the sensing history relevant for target
prediction, while $\bfH_{c,k}$ is the currently observed communication
channel. Both are available before selecting the block-input law
\begin{equation}
    a_k
    =
    P_{\bfX_k|\cF_{k-1}}(\cdot|\cF_{k-1}),
\end{equation}
which induces the covariance $\bfQ(a_k)\in\cU$ through
\eqref{eq:conditional_covariance}.

For a fixed causal policy $\pi$, the PCRB matrix
$\bfJ_k^\pi$ in~\eqref{eq:PCRB_recursion} is deterministic and therefore is
not a component of the controlled information state. The policy affects the
PCRB through the policy-averaged data-information sequence
$\{\bfM_k^\pi\}$.

Given $\bfZ_k$ and the selected block-input law $a_k$, the Bayesian filtering
recursion determines the distribution $\Pi_{k+1}$, while
$\bfH_{c,k+1}$ is generated independently according to the communication
channel law. Hence,
\begin{equation}
    P(d\bfz_{k+1}|
      \bfZ_{1:k},a_{1:k})
    =
    P(d\bfz_{k+1}|\bfZ_k,a_k),
    \label{eq:controlled_markov_property}
\end{equation}
and $\{\bfZ_k\}$ is a fully observed controlled Markov process.

For an admissible state--action pair $(\bfz,a)$, define the Gaussian
communication upper reward
\begin{equation}
    r_{\rm G}(\bfz,a)
    \triangleq
    \log_2\det\!\left(
        \bfI
        +
        \sigma_c^{-2}
        \bfH_c
        \bfQ(a)
        \bfH_c^{\herm}
    \right),
    \qquad
    \bfz=(\Pi,\bfH_c).
    \label{eq:one_step_reward}
\end{equation}
This reward results from revealing the adaptive covariance action to the
communication receiver as genie side information and then applying the
Gaussian maximum-entropy bound for the MIMO channel. Accordingly, define
\begin{equation}
    R_{\rm G}(\pi)
    \triangleq
    \liminf_{N\to\infty}
    \frac{1}{N}
    \sum_{k=1}^{N}
    \E^\pi
    \left[
        r_{\rm G}(\bfZ_k,a_k)
    \right].
    \label{eq:gaussian_average_reward}
\end{equation}
Let
$\bfQ_{1:N}\triangleq(\bfQ_{\bfx,1},\ldots,\bfQ_{\bfx,N})$ and define the
conditional spectral sup-mutual information rate
\begin{equation}
    \overline I_{Q,L}(\pi)
    \triangleq
    \rmp\text{-}\limsup_{N\to\infty}
    \frac{1}{NL}
    i\!\left(
        \bfQ_{1:N};\bfY_{c,1:N}
        \,\middle|\,
        \bfH_{c,1:N}
    \right).
    \label{eq:covariance_information_rate}
\end{equation}
Here $\rmp\text{-}\limsup$ denotes the limit superior in probability.
The conditional information-density chain rule and the Gaussian
maximum-entropy bound give
\begin{equation}
    \underline R(\pi)
    \leq
    R_{\rm G}(\pi)
    +
    \overline I_{Q,L}(\pi).
    \label{eq:covariance_genie_bound}
\end{equation}
A proof is given in Appendix~\ref{app:large_block_gaussian}.

We therefore consider the covariance-based scalarized problem
\begin{equation}
    J_{\lambda}^{\rm out}
    \triangleq
    \sup_{\pi\in\Pi_{\rm causal}}
    \left\{
        R_{\rm G}(\pi)
        -
        \lambda\overline D(\pi)
    \right\},
    \qquad
    \lambda\geq0,
    \label{eq:outer_scalarized_problem}
\end{equation}
and define
\begin{equation}
    \epsilon_{Q,L}
    \triangleq
    \max\!\left\{
        0,\,
        \sup_{\pi\in\Pi_{\rm causal}}
        \overline I_{Q,L}(\pi)
    \right\}.
    \label{eq:covariance_information_uniform_bound}
\end{equation}
By~\eqref{eq:covariance_genie_bound},
\begin{equation}
    J_{\lambda}^{\star}
    \leq
    J_{\lambda}^{\rm out}
    +
    \epsilon_{Q,L}.
    \label{eq:finite_block_outer_bound}
\end{equation}

\subsection{Sufficiency of Stationarity}

We now state the regularity conditions required for the stationary-policy
reduction of $J_{\lambda}^{\rm out}$. Let $\sfZ$ denote the information-state space and $\sfA$ the space
of admissible block-input laws.

\begin{assumption}[Controlled-Markov regularity]
\label{ass:mdp_regular}
The following conditions hold:
\begin{enumerate}
    \item The state and action spaces are Polish, the admissible
    state--action graph is Borel, and the controlled transition kernel
    $P(d\bfz'|\bfz,a)$ is weakly continuous.

    \item For every admissible causal policy, the expected empirical
    occupation measures are tight, and the set of invariant occupation
    measures is weakly compact.

    \item The functions $r_{\rm G}(\bfz,a)$ and
    $\bfJ_D(\Pi,\bfQ(a))$ are continuous on the admissible graph. Moreover,
    there exists $\bar g<\infty$ such that
    \begin{equation}
        {\bf0}
        \preceq
        \bfJ_D(\Pi,\bfQ(a))
        \preceq
        {\bar g} \bfI
    \end{equation}
    for every admissible state--action pair, while $r_{\rm G}$ satisfies the
    boundedness or uniform-integrability condition required for its average to be
    preserved under weak convergence of occupation measures.

    \item $\bfQ_w\succ\mathbf 0$, $\bfJ_0\succ\mathbf 0$, and $\bfF$ is
    Schur stable.

    \item For every admissible stationary randomized Markov policy, the
    long-run quantities $R_{\rm G}$ and $\overline D$ under the prescribed
    initial information-state distribution coincide with those under any
    invariant initialization of that policy.
\end{enumerate}
\definitionend
\end{assumption}

\begin{prop}[Sufficiency of Stationary Policies]
\label{prop:stationary_suffices_conf}
Under Assumption~\ref{ass:mdp_regular}, for every admissible causal policy
$\pi$ there exists a stationary randomized Markov policy $\varphi$ such that
\begin{align}
    R_{\rm G}(\varphi)
    &\geq
    R_{\rm G}(\pi),
    \label{eq:stationary_dominance_rate}\\
    \overline{D}(\varphi)
    &\leq
    \overline{D}(\pi).
    \label{eq:stationary_dominance_distortion}
\end{align}
Consequently, for every $\lambda\geq0$,
\begin{align}
    J_{\lambda}^{\rm out}
    &=
    \sup_{\varphi\in\Pi_{\rm stat}}
    \left\{
        R_{\rm G}(\varphi)
        -
        \lambda \overline{D}(\varphi)
    \right\},
    \label{eq:stationary_outer_problem}
\end{align}
where $\Pi_{\rm stat}$ denotes the class of stationary randomized Markov
policies on the information state $\bfZ_k$. Moreover, the supremum on the
right-hand side is attained.
\statementend
\end{prop}

\begin{IEEEproof}[Proof sketch]
Fix an admissible causal policy $\pi$ and select a subsequence realizing the
superior limit in $\overline{D}(\pi)$. Tightness of the corresponding expected empirical
occupation measures yields a further subsequence converging weakly to an
invariant occupation measure $\mu$, whose disintegration induces a stationary
randomized Markov policy $\varphi$. Let
$\bfM_\mu\triangleq
\int\bfJ_D(\Pi,\bfQ(a))\,\mu(d\bfz,da)$.

Averaging the deterministic PCRB recursion along the selected subsequence and
using the monotonicity and matrix concavity of the prediction-information map
shows that the limiting average information matrix is dominated by the unique
fixed point associated with $\bfM_\mu$. Convexity and order monotonicity of
$\tr{\bfS\bfJ^{-1}\bfS^T}$ then yield $\overline{D}(\varphi)\leq \overline{D}(\pi)$. Weak
convergence along the same subsequence gives
$R_{\rm G}(\varphi)\geq R_{\rm G}(\pi)$, while compactness of the invariant
occupation-measure set and continuity of the stationary objective establish
attainment. The supporting technical lemmas and the complete proof are 
given in Appendices~\ref{app:stationary_lemmas}
and~\ref{app:proof_stationary_suffices}, respectively.
\end{IEEEproof}

\subsection{Stationary Operation and Outer-Bound Metrics}
\label{subsec:stationary_operation}

Let $\pi\in\Pi_{\rm stat}$ be a stationary randomized Markov policy and let
$\mu_\pi$ denote an invariant state--action occupation measure induced by
$\pi$. Under initialization with the state marginal of $\mu_\pi$, the
communication upper reward is
\begin{equation}
    R_{\rm G}(\pi)
    =
    \E_{\mu_\pi}
    \left[
        r_{\rm G}(\bfZ,a)
    \right].
    \label{eq:ergodic_rate_conf}
\end{equation}
The corresponding average sensing data-information matrix is
\begin{equation}
    \bar{\bfM}_{\pi}
    \triangleq
    \E_{\mu_\pi}
    \left[
        \bfJ_D(\Pi,\bfQ(a))
    \right].
    \label{eq:stationary_data_information}
\end{equation}
The policy-averaged PCRB recursion then has a constant information contribution,
and its steady-state information matrix $\bfJ_{\pi}$ is the unique
positive-definite solution of
\begin{equation}
    \bfJ_{\pi}
    =
    \left(
        \bfQ_w
        +
        \bfF\bfJ_{\pi}^{-1}\bfF^T
    \right)^{-1}
    +
    \bar{\bfM}_{\pi}.
    \label{eq:stationary_pcrb_fixed_point}
\end{equation}
The corresponding stationary PCRB distortion is
\begin{equation}
    \overline{D}(\pi)
    =
    \tr{
        \bfS
        \bfJ_{\pi}^{-1}
        \bfS^T
    }.
    \label{eq:steady_state_distortion_conf}
\end{equation}
By Assumption~\ref{ass:mdp_regular}, these stationary values also equal
the corresponding long-run quantities under the prescribed initialization.

\section{Gaussian Covariance-Based Outer Bound}
\label{sec:gaussian_outer_bound}

For a fixed target parameter over a single large block, Gaussian signaling and
a covariance-based characterization of the communication--estimation trade-off
were established in~\cite{yilmaz2026joint}. In the dynamic
tracking setting considered here, however, 
Proposition~\ref{prop:stationary_suffices_conf} reduces the covariance-based
scalarized problem to stationary randomized Markov policies. It remains to justify
whether the within-block signaling law can also be restricted to Gaussian
signaling. Such a restriction is not exact at finite block length: two
signaling laws with the same nominal covariance can generate different
realizations of the empirical covariance and, consequently, different
posterior-belief transitions. We therefore establish the Gaussian reduction in
the large-block regime.

\subsection{Finite-Block Sufficient Statistics and Large-Block Limit}

Stacking the $L$ sensing observations in
\eqref{eq:sensing_model_conf} gives
\begin{equation}
    \bfY_{s,k}
    =
    \bfH_s(\bfs_k)\bfX_k+\bfZ_{s,k}.
    \label{eq:block_sensing_gaussian_section}
\end{equation}
For the large-block comparison, we consider admissible block-signaling
families satisfying
\begin{equation}
    \widehat{\bfQ}_{k,L}
    \xrightarrow[L\to\infty]{\mathbb P}
    \bfQ_k,
    \label{eq:large_block_empirical_covariance}
\end{equation}
uniformly on compact subsets of the admissible state--covariance graph, where
$\bfQ_k\in\cU$ is the covariance selected by the stationary covariance policy.
All optimization level asymptotic statements below are restricted to this
same admissible large-block signaling class.
The per channel use sensing-noise variance in the asymptotic family is allowed
to depend on $L$ and is assumed to satisfy
\begin{equation}
    \frac{\sigma_{s,L}^{2}}{L}
    \to
    \bar{\sigma}_s^{2},
    \qquad
    \bar{\sigma}_s^{2}\in(0,\infty),
    \label{eq:large_block_sensing_normalization}
\end{equation}
as $L\to\infty$. This scaling yields a nondegenerate limiting sensing
experiment. It is imposed only for the large-block asymptotic analysis; the
finite block model in Section~\ref{sec:system_model} retains the fixed
per channel use variance $\sigma_s^2$. In that model, the sensing data
information in~\eqref{eq:data_fim_covariance_form} scales as
$L/\sigma_s^2$.

\begin{lemma}[Finite-block Gram sufficiency]
\label{lem:block_gram_sufficiency}
For a fixed transmitted block $\bfX_k$, the posterior distribution of $\bfs_k$
generated by~\eqref{eq:block_sensing_gaussian_section} depends on
$(\bfY_{s,k},\bfX_k)$ only through
\begin{equation}
    \bfT_{k,L}
    \triangleq
    \bfY_{s,k}\bfX_k^{\herm},
    \qquad
    \bfG_{k,L}
    \triangleq
    \bfX_k\bfX_k^{\herm}.
    \label{eq:block_sufficient_statistics}
\end{equation}
Moreover, conditional on the true target state and on $\bfX_k$, the
distribution of the random posterior depends on $\bfX_k$ only through
$\bfG_{k,L}$.
\statementend
\end{lemma}

\begin{IEEEproof}
For fixed $\bfX_k$, the sensing likelihood is proportional to
\begin{equation}
    \exp\!\left(
        -\frac{1}{\sigma_{s,L}^{2}}
        \left\|
            \bfY_{s,k}-\bfH_s(\bfs_k)\bfX_k
        \right\|_F^2
    \right).
\end{equation}
Expanding the squared norm gives
\begin{align}
    &
    \left\|
        \bfY_{s,k}-\bfH_s(\bfs_k)\bfX_k
    \right\|_F^2
    \notag\\
    &=
    \tr{\bfY_{s,k}\bfY_{s,k}^{\herm}}
    -2\Re\!
    \left\{
        \tr{
            \bfH_s^{\herm}(\bfs_k)
            \bfY_{s,k}\bfX_k^{\herm}
        }
    \right\}
    \notag\\
    &\quad
    +
    \tr{
        \bfH_s^{\herm}(\bfs_k)
        \bfH_s(\bfs_k)
        \bfX_k\bfX_k^{\herm}
    }.
    \label{eq:block_likelihood_expansion}
\end{align}
The first term is independent of $\bfs_k$ and cancels in the Bayesian
normalization. Hence the posterior depends on the data only through
$(\bfT_{k,L},\bfG_{k,L})$. Under the true state $\bfs_k^0$,
\begin{equation}
    \bfT_{k,L}
    =
    \bfH_s(\bfs_k^0)\bfG_{k,L}
    +
    \bfZ_{s,k}\bfX_k^{\herm},
\end{equation}
and, conditional on $\bfX_k$, each row of
$\bfZ_{s,k}\bfX_k^{\herm}$ is circular Gaussian with covariance proportional
to $\bfG_{k,L}$. The conditional distribution of the random posterior
therefore depends on $\bfX_k$ only through $\bfG_{k,L}$.
\end{IEEEproof}

To obtain the covariance-only large-block limit, normalize the sufficient
statistics as
\begin{equation}
    \widehat{\bfQ}_{k,L}
    =
    \frac{1}{L}\bfG_{k,L},
    \qquad
    \overline{\bfT}_{k,L}
    =
    \frac{1}{L}\bfT_{k,L}.
    \label{eq:normalized_sufficient_statistics}
\end{equation}
Under the true state,
\begin{equation}
    \overline{\bfT}_{k,L}
    =
    \bfH_s(\bfs_k^0)\widehat{\bfQ}_{k,L}
    +
    \bfW_{k,L},
    \qquad
    \bfW_{k,L}
    \triangleq
    \frac{1}{L}\bfZ_{s,k}\bfX_k^{\herm}.
    \label{eq:normalized_cross_statistic}
\end{equation}
Using~\eqref{eq:large_block_empirical_covariance} and
\eqref{eq:large_block_sensing_normalization}, the conditional distribution of
these normalized sufficient statistics converges to a limit determined only by
the selected covariance $\bfQ_k$. Under continuity of the Bayesian update, all
admissible within-block signaling laws satisfying
\eqref{eq:large_block_empirical_covariance} and implementing the same
covariance therefore induce the same limiting belief-transition kernel.

We formalize the stationary consequence under the following regularity
condition.

\begin{assumption}[Stationary large-block limit]
\label{ass:large_block_stationary}
The Bayesian posterior update is continuous, on the admissible set, as a
function of the predictive belief and the normalized sufficient statistics.
For the admissible stationary large-block signaling families, the invariant
state--covariance occupation measures are uniformly tight in $L$, and the
convergence of the induced controlled information-state kernels to the
covariance-only large-block limit established above is uniform on compact
subsets of the admissible state--covariance graph. For every stationary
randomized covariance policy arising either directly or as a weak limit of
such occupation measures, the limiting controlled process has a unique
invariant occupation measure.
\definitionend
\end{assumption}

\subsection{Asymptotic Gaussian Reduction}

For a selected covariance $\bfQ_k$, the corresponding conditionally i.i.d.
Gaussian implementation is
\begin{equation}
    \bfx_{k,\ell}
    \sim
    \mathcal{CN}(\mathbf 0,\bfQ_k),
    \qquad
    \ell=1,\ldots,L.
\end{equation}

For an implementation $i\in\{A,G\}$, let
\begin{equation}
    R_{{\rm com},L}^{i}
    \triangleq
    \frac{1}{L}
    I_i\!\left(
        \bfX_k;\bfY_{c,k}
        \,\middle|\,
        \bfH_{c,k},\bfQ_k
    \right)
    \label{eq:genie_conditioned_stationary_rate}
\end{equation}
denote the corresponding block mutual information under its stationary law,
where the mutual information is evaluated under implementation $i$.

\begin{prop}[Large-block Gaussian signaling optimality]
\label{prop:large_block_gaussian}
Suppose Assumption~\ref{ass:mdp_regular},
Assumption~\ref{ass:large_block_stationary}, and
\eqref{eq:large_block_empirical_covariance} hold. Fix a stationary randomized
covariance policy and compare an arbitrary admissible within-block signaling
implementation, denoted by $A$, with its conditionally i.i.d. Gaussian
implementation, denoted by $G$, using the same covariance policy. Then
\begin{equation}
    D_{{\rm PCRB},L}^{G}
    -
    D_{{\rm PCRB},L}^{A}
    \longrightarrow 0,
    \qquad
    L\to\infty,
    \label{eq:large_block_pcrb_equality}
\end{equation}
and the corresponding genie-conditioned stationary communication rates satisfy
\begin{equation}
    \liminf_{L\to\infty}
    \left(
        R_{{\rm com},L}^{G}
        -
        R_{{\rm com},L}^{A}
    \right)
    \geq 0.
    \label{eq:large_block_rate_dominance}
\end{equation}
Hence, for any fixed stationary covariance policy, conditionally Gaussian
signaling is asymptotically no worse for the corresponding covariance-based 
scalarized objective.
\statementend
\end{prop}

\begin{IEEEproof}[Proof sketch]
For a fixed stationary covariance policy, the arbitrary and Gaussian
implementations induce the same limiting invariant state--covariance
occupation measure under Assumption~\ref{ass:large_block_stationary}.
Since $\bfJ_D(\Pi,\bfQ)$ depends on the signaling action only through its
covariance, continuity of the stationary PCRB fixed point yields
\eqref{eq:large_block_pcrb_equality}. For communication, the Gaussian
maximum-entropy property together with concavity of $\log\det$ upper-bounds
the per-channel-use mutual information of any admissible length-$L$ block law
by the Gaussian reward associated with its average covariance, with equality
for conditionally i.i.d. Gaussian signaling. This gives
\eqref{eq:large_block_rate_dominance}. The complete proof is given in
Appendix~\ref{app:large_block_gaussian}.
\end{IEEEproof}

Proposition~\ref{prop:large_block_gaussian} concerns a fixed stationary
covariance policy. To justify the Gaussian restriction at the optimization
level, let $V_{\lambda,L}^{A}$ denote the optimum of the stationary
covariance-control outer problem over all admissible block-signaling
implementations satisfying~\eqref{eq:large_block_empirical_covariance}, and
let $V_{\lambda,L}^{G}$ denote the corresponding optimum restricted to
conditionally i.i.d. Gaussian signaling.

To connect this covariance reduction to the original Verd\'u--Han objective,
we impose the following large-block condition.
\begin{assumption}[Vanishing covariance information rate]
\label{ass:vanishing_covariance_information}
Let $\epsilon_{Q,L}^{\rm LB}$ denote the quantity in
\eqref{eq:covariance_information_uniform_bound} with the supremum restricted
to the admissible large-block signaling families satisfying
\eqref{eq:large_block_empirical_covariance}. Then
\begin{equation}
    \epsilon_{Q,L}^{\rm LB}
    =
    o(1),
    \qquad
    L\to\infty.
    \label{eq:vanishing_covariance_information}
\end{equation}
\end{assumption}
For regular parametric families of fixed finite dimension, standard
parametric-information bounds give $O(\log L)$ information growth over a
block and hence $O(\log L/L)$ per channel use
\cite{ClarkeBarron1990,HausslerOpper1997}.

\begin{prop}[Optimization-level Gaussian reduction]
\label{prop:gaussian_outer_bound_conf}
Under the assumptions of
Proposition~\ref{prop:large_block_gaussian}, for every $\lambda\geq0$,
\begin{equation}
    \limsup_{L\to\infty}
    \left(
        V_{\lambda,L}^{A}
        -
        V_{\lambda,L}^{G}
    \right)
    \leq 0.
    \label{eq:optimized_gaussian_limsup}
\end{equation}
Since
$V_{\lambda,L}^{G}\leq V_{\lambda,L}^{A}$
for every $L$, it follows that
\begin{equation}
    V_{\lambda,L}^{A}
    =
    V_{\lambda,L}^{G}
    +
    o(1),
    \qquad
    L\to\infty.
    \label{eq:optimized_gaussian_equality}
\end{equation}
If, in addition, Assumption~\ref{ass:vanishing_covariance_information} holds and
$J_{\lambda,L}^{\star,\rm LB}$ denotes the scalarized
communication--PCRB problem in~\eqref{eq:J_star_conf} restricted to the same
admissible large-block signaling class, then
\begin{equation}
    J_{\lambda,L}^{\star,\rm LB}
    \leq
    V_{\lambda,L}^{G}
    +
    o(1),
    \qquad
    L\to\infty.
    \label{eq:gaussian_outer_bound_conf}
\end{equation}
\statementend
\end{prop}

\begin{IEEEproof}[Proof sketch]
Select nearly optimal admissible stationary policies along a sequence
realizing the superior limit in~\eqref{eq:optimized_gaussian_limsup}.
Uniform tightness yields a limiting invariant occupation measure and hence a
limiting stationary covariance policy. Its conditionally Gaussian
implementation converges to the same covariance-controlled invariant limit,
so continuity of the communication reward and stationary PCRB yields
\eqref{eq:optimized_gaussian_limsup}. Finally, applying the argument of~\eqref{eq:finite_block_outer_bound} to
the same admissible large-block signaling class and using
Proposition~\ref{prop:stationary_suffices_conf} gives
\begin{equation}
    J_{\lambda,L}^{\star,\rm LB}
    \leq
    V_{\lambda,L}^{A}
    +
    \epsilon_{Q,L}^{\rm LB}.
\end{equation}
Together with~\eqref{eq:optimized_gaussian_equality} and
Assumption~\ref{ass:vanishing_covariance_information}, this gives
\eqref{eq:gaussian_outer_bound_conf}. The complete proof is provided in
Appendix~\ref{app:large_block_gaussian}.
\end{IEEEproof}

\begin{rem}[Scope of the Gaussian reduction]
\label{rem:gaussian_scope}
The Gaussian reduction is asymptotic in the tracking-block length. At finite
$L$, Lemma~\ref{lem:block_gram_sufficiency} establishes dependence on the
realized Gram matrix $\bfX_k\bfX_k^{\herm}$; it does not assert that two
same-covariance signaling laws induce identical posterior beliefs. The
covariance-only belief transition, and hence the asymptotic preservation of the
stationary PCRB under Gaussian replacement, follows only after the large-block
limit.
\end{rem}

Under the stated assumptions, 
Proposition~\ref{prop:gaussian_outer_bound_conf} reduces the
large-block outer problem asymptotically to covariance control, but the resulting
optimization is generally not a static convex program. For a fixed
communication channel, $r_{\rm G}(\bfz,\bfQ)$ is concave in $\bfQ$ and the
covariance action set $\cU$ is convex. The difficulty is dynamic: each
covariance action affects the sensing experiment, the future predictive
belief, and hence the invariant occupation measure. Moreover, the stationary
PCRB depends nonlinearly on the policy-induced average data information
through~\eqref{eq:stationary_pcrb_fixed_point}. The problem therefore remains
an average-reward control problem on the information state
$(\Pi_k,\bfH_{c,k})$.

\section{Numerical Examples}
\label{sec:simulations}

We numerically evaluate the covariance-control formulation developed above by
comparing an EKF-based achievable frontier with a PCRB-based Gaussian
outer frontier. Rates are reported in bits per channel use. Throughout this
section, $\bfQ_k$ denotes the covariance action $\bfQ_{\bfx,k}$ selected at
tracking time $k$.

The target follows the Gauss--Markov model in~\eqref{eq:state_evolution}, and
the sensing model is given by~\eqref{eq:sensing_model_conf} with
$\lambda_c=10$ cm and $\beta=200$. We use uniform linear arrays with steering
vector
\begin{equation}
    [\bfa(\theta)]_i
    =
    e^{j2\pi \frac{\Delta}{\lambda_c}(i-1)\sin(\theta)},
    \qquad
    i=1,\ldots,N_t,
\end{equation}
where $\Delta=\lambda_c/2$.

Let $\bfu_{\rm tx}(\theta)$ and $\bfu_{\rm rx}(\theta)$ denote the unit-norm transmit and receive vectors obtained by normalizing $\bfa(\theta)$. For simplicity, the numerical example considers a deterministic far-field line of sight communication channel, with the BS and communication user arrays seeing each other at $\theta_c=45^\circ$ relative to their respective boresight directions
\begin{equation}
    \bfH_c
    =
    \sqrt{N_{r,c}N_t}\,g\,
    \bfu_{\mathrm{rx}}(\theta_c)
    \bfu_{\mathrm{tx}}^{\herm}(\theta_c),
    \qquad
    |g|^2=\sigma_h^2 .
\end{equation}
This is a degenerate i.i.d. channel. Since $\bfH_c$ is fixed in the numerical
experiments, it is not a variable component of the policy state. For
convenience, let $r(\bfH,\bfQ)$ denote the Gaussian communication reward in
\eqref{eq:one_step_reward} evaluated at channel $\bfH$ and covariance $\bfQ$.
For the rank-one channel above,
\begin{equation}
    r(\bfH_c,\bfQ)
    =
    \log_2\left(
        1+
        \frac{N_{r,c}N_t\sigma_h^2}{\sigma_c^2}
        \bfu_{\mathrm{tx}}^{\herm}(\theta_c)
        \bfQ
        \bfu_{\mathrm{tx}}(\theta_c)
    \right).
\end{equation}

Unless otherwise stated, the numerical parameters are listed in
Table~\ref{tab:sim_params}. The communication reward uses noise variance
$\sigma_c^2$. For sensing, define the block normalized noise variance
\begin{equation}
    \sigma_{s,\mathrm{blk}}^2
    \triangleq
    \frac{\sigma_s^2}{L}.
    \label{eq:block_noise_numerical}
\end{equation}
In the large-block regime,
$\widehat{\bfQ}_{k,L}\xrightarrow{\mathbb P}\bfQ_k$, and an
information-equivalent normalized sensing statistic may be written as
\begin{equation}
    \widetilde{\bfY}_{s,k}
    =
    \bfH_s(\bfs_k)\bfQ_k^{1/2}
    +
    \widetilde{\bfZ}_{s,k},
    \qquad
    [\widetilde{\bfZ}_{s,k}]_{ij}
    \sim
    \mathcal{CN}(0,\sigma_{s,\mathrm{blk}}^2).
    \label{eq:block_sensing_numerical}
\end{equation}
Here $\bfQ_k^{1/2}$ denotes the Hermitian positive semidefinite square root
of $\bfQ_k$. Its data-information coefficient satisfies
$2/\sigma_{s,\mathrm{blk}}^2=2L/\sigma_s^2$, consistently with
\eqref{eq:data_fim_covariance_form}. Accordingly, the numerical sensing experiment is parameterized by $P/\sigma_{s,\mathrm{blk}}^2=LP/\sigma_s^2$, so $L$ and $\sigma_s^2$
need not be specified separately. The target is initialized with radial
speed $\dot d_0=14$ m/s, corresponding to $50$ km/h.

\begin{table}[t]
    \centering
    \caption{Simulation Parameters}
    \label{tab:sim_params}
    \begin{tabular}{|l|c|}
        \hline
        Parameter & Value \\
        \hline
        Time step duration ($\Delta t$) & $0.05$ s \\
        Burn-in steps ($T_{\rm burn-in}$) & $20$ \\
        Learning steps ($T_{\rm learn}$) & $40$ \\
        Evaluation steps ($T_{\rm eval}$) & $80$ \\
        Communication transmit SNR ($P/\sigma_c^2$) & $20$ dB \\
        Communication noise variance ($\sigma_c^2$) & $1$ \\
        Block sensing SNR ($P/\sigma_{s,\mathrm{blk}}^2$) & $20$ dB \\
        MPC horizon ($H$) & $6$ \\
        Rollouts per action ($M$) & $8$ \\
        Monte Carlo trials ($N_{\rm MC}$) & $20{,}000$ \\
        Transmit antennas ($N_t$) & $8$ \\
        Communication receive antennas ($N_{r,c}$) & $8$ \\
        Sensing receive antennas ($N_{r,s}$) & $8$ \\
        Process-noise covariance matrix ($\bfQ_w$) & $10^{-5}\bfI_{4}$ \\
        \hline
    \end{tabular}
\end{table}

\subsection{Covariance Policy Classes}

The EKF-based achievable frontier uses a structured covariance family to
obtain a concrete implementable signaling and tracking scheme, whereas the
PCRB-based Gaussian outer frontier optimizes over the full covariance action
set $\cU$.
For the EKF-based achievable frontier, we use the following structured
covariance family. Given the
predicted target angle $\hat\theta_{k|k-1}$ and a sensing-beam offset
$\Delta\theta$, define
\begin{equation}
    \bfu_s=\bfu_{\mathrm{tx}}(\hat\theta_{k|k-1}+\Delta\theta),
    \qquad
    \bfu_c=\bfu_{\mathrm{tx}}(\theta_c).
\end{equation}
For an action $a=(\rho_{\rm comm},\Delta\theta,\gamma)$, the transmit covariance
is
\begin{align}
    \bfQ(a;\hat{\bfs}_{k|k-1})
    &=
    P_s\bfu_s\bfu_s^{\herm}
    +
    P_c\bfu_c\bfu_c^{\herm}
    +
    \frac{P_{\rm iso}}{N_t}\bfI, \\
    P_{\rm iso} &= \gamma P, \\
    P_c &= \rho_{\rm comm}(1-\gamma)P, \\
    P_s &= (1-\rho_{\rm comm})(1-\gamma)P .
\end{align}
Here $\rho_{\rm comm}$ controls the power split between communication and
sensing, $\Delta\theta$ controls the sensing beam direction, and $\gamma$
controls the isotropic component. In the reported rank-one experiments,
$\gamma=0$.

\subsection{Achievable EKF-Based Scheme}

For each covariance action $\bfQ_k$, the transmitter uses conditionally
i.i.d. Gaussian signaling within the block,
\begin{equation}
    \bfx_{k,\ell}
    \sim
    \mathcal{CN}(\mathbf 0,\bfQ_k),
    \qquad
    \ell=1,\ldots,L.
\end{equation}
The EKF processes the block observation in
\eqref{eq:block_sensing_numerical}; its implementation is given in
Appendix~\ref{app:EKF}, following~\cite{simon2006optimal}. The distortion is the actual Monte Carlo squared
tracking error in angle and distance.

The deployed policies use the predictive-belief feature
\begin{equation}
    \scalebox{0.96}{$\displaystyle
    \phi_k
    =
    \left[
        \hat{\bfs}_{k|k-1}^{T},
        \log\operatorname{diag}(\bfP_{k|k-1})^{T},
        \operatorname{offdiag}\!\left(
            \operatorname{Corr}(\bfP_{k|k-1})
        \right)^{T}
    \right]^{T},
    $}
\end{equation}
where $\hat{\bfs}_{k|k-1}$ and $\bfP_{k|k-1}$ are the predicted state estimate
and prediction-error covariance, respectively. The operators
$\operatorname{Corr}(\cdot)$ and $\operatorname{offdiag}(\cdot)$ denote the
associated correlation matrix and its stacked off-diagonal entries.

The achievable distortion is
\begin{align}
    \widehat D_{\rm ach}
    &=
    \frac{1}{N_{\rm MC}T_{\rm eval}}
    \sum_{n=1}^{N_{\rm MC}}
    \sum_{k=1}^{T_{\rm eval}}
    \Big[
    {\rm wrap}(\hat\theta_{n,k}-\theta_{n,k})^2
    \notag\\
    &\qquad\qquad\qquad\qquad
    +
    (\hat d_{n,k}-d_{n,k})^2
    \Big],
    \label{eq:achievable_distortion_sim}
\end{align}
where ${\rm wrap}(x)$ maps angle errors to $[-\pi,\pi)$.

\subsection{PCRB-Based Gaussian Outer Frontier}

For a deployed stationary covariance policy $\pi$, the PCRB is evaluated at
the policy level rather than separately for individual Monte Carlo
realizations. Let $\Pi_{n,k}$ and $\bfQ_{n,k}$ denote the predictive belief
and covariance action in trial $n$ at evaluation time $k$. The
policy-averaged data-information matrix is estimated as
\begin{equation}
    \widehat{\bfM}_{\pi}
    =
    \frac{1}{N_{\rm MC}T_{\rm eval}}
    \sum_{n=1}^{N_{\rm MC}}
    \sum_{k=1}^{T_{\rm eval}}
    \bfJ_D(\Pi_{n,k},\bfQ_{n,k}),
    \label{eq:MC_policy_information}
\end{equation}
where $\bfJ_D(\Pi_{n,k},\bfQ_{n,k})$ denotes the conditional data-information
matrix averaged with respect to the predictive belief. The theoretical outer 
bound is defined in terms of the Bayesian predictive belief and is independent 
of any particular estimator. For numerical evaluation, we approximate this belief
by a Gaussian distribution with mean $\hat{\bfs}_{n,k|k-1}$ and covariance
$\bfP_{n,k|k-1}$ and evaluate the resulting angle--range expectation using
cubature. We also evaluated the
predictive belief using a particle filter and observed no appreciable change
in the resulting PCRB curves for the scenarios considered here. More generally,
for models that yield strongly non-Gaussian or multimodal predictive
distributions, particle-based belief representations may provide a more
accurate approximation.

Given $\widehat{\bfM}_{\pi}$, a single deterministic stationary PCRB
information matrix is obtained as the positive-definite solution of
\begin{equation}
    \widehat{\bfJ}_{\pi}
    =
    \left(
        \bfQ_w
        +
        \bfF
        \widehat{\bfJ}_{\pi}^{-1}
        \bfF^T
    \right)^{-1}
    +
    \widehat{\bfM}_{\pi}.
    \label{eq:MC_stationary_PCRB_fixed_point}
\end{equation}
The corresponding PCRB-based tracking distortion is
\begin{equation}
    \widehat D_{\rm PCRB}
    =
    \tr{
        \bfS
        \widehat{\bfJ}_{\pi}^{-1}
        \bfS^T
    }
    =
    [\widehat{\bfJ}_{\pi}^{-1}]_{1,1}
    +
    [\widehat{\bfJ}_{\pi}^{-1}]_{3,3}.
    \label{eq:pcrb_distortion_sim}
\end{equation}
Thus, the Monte Carlo realizations are used to evaluate the joint
policy-induced average in~\eqref{eq:MC_policy_information}, after which the
nonlinear PCRB fixed-point equation is solved once for the deployed policy.

The outer policy uses the same predictive-belief feature $\phi_k$ defined
above. The PCRB information matrix $\widehat{\bfJ}_{\pi}$ is a policy-level
performance quantity and is not included in the policy state.

The corresponding policy-average communication reward is
\begin{equation}
    \widehat R
    =
    \frac{1}{N_{\rm MC}T_{\rm eval}}
    \sum_{n=1}^{N_{\rm MC}}
    \sum_{k=1}^{T_{\rm eval}}
    r(\bfH_c,\bfQ_{n,k}).
\end{equation}
Each outer-frontier point is the pair
$(\widehat R,\widehat D_{\rm PCRB})$
associated with the deployed stationary covariance policy.

\subsection{MPC-Distilled Stationary Policies}

The stationary covariance policies used in the numerical evaluation are
obtained by an MPC-distillation procedure~\cite[Ch.~2]{MPC}. For each
scalarization parameter $\lambda$, a finite-horizon MPC teacher is run along
a training trajectory. The action class and sensing metric depend on the
frontier being computed. For the EKF-based achievable frontier, the teacher
uses the structured covariance family and the realized EKF tracking error. For
the PCRB-based Gaussian outer frontier, the teacher searches over the full
covariance action set $\cU$ and evaluates every candidate covariance sequence
using the policy-averaged PCRB construction described below.

For the PCRB-based outer frontier, each candidate covariance sequence is
evaluated over $M$ predictive-belief rollouts of horizon $H$. The teacher
averages the corresponding data-information matrices and communication
rewards, solves one stationary PCRB fixed-point equation for the resulting
average data information, and selects the sequence maximizing the scalarized
reward $\widehat R_i-\lambda\widehat D_i$. Only the first covariance action is
applied according to the receding-horizon principle. The complete teacher
calculation is given in Algorithm~\ref{alg:nn_policy}.

Teacher labels are collected over $T_{\rm learn}=40$ learning steps. At each
labeled state, the predictive-belief feature $\phi_i$ and the corresponding
optimal full-covariance action $\bfQ_i^\star\in\cU$ produced by the MPC teacher
are stored. After a $20$-step burn-in, the deployed policy standardizes the
current predictive-belief feature, selects the nearest stored feature, and
applies its associated covariance action. Since this feature-to-covariance
mapping is time invariant, the deployed policy is stationary. The policy is
then evaluated for $T_{\rm eval}=80$ steps. The complete training and
deployment procedure is summarized in Algorithm~\ref{alg:nn_policy}.

\begin{algorithm}
    \caption{MPC-Distilled Stationary Covariance Policy}
    \label{alg:nn_policy}
    \begin{algorithmic}[1]
        \REQUIRE Scalarization parameter $\lambda$, covariance action set
        $\cU$, rollout horizon $H$, number of rollout samples $M$, and
        learning length $T_{\rm learn}$.
        \ENSURE Stationary covariance policy $\pi_{\rm NN}$.

        \STATE Initialize the teacher dataset $\cD\leftarrow\emptyset$.

        \STATE \textit{Phase 1: MPC teacher labeling}
        \FOR{training step $i=1$ to $T_{\rm learn}$}
            \STATE Form the predictive-belief feature $\phi_i$.

            \STATE For each candidate covariance sequence
            $\{\bfQ_{i+\ell}\}_{\ell=0}^{H-1}$, compute
            \[
                \widehat{\bfM}_{i}
                =
                \frac{1}{MH}
                \sum_{m=1}^{M}
                \sum_{\ell=0}^{H-1}
                \bfJ_D
                \left(
                    \Pi_{i+\ell}^{(m)},
                    \bfQ_{i+\ell}
                \right)
            \]
            and
            \[
                \widehat R_i
                =
                \frac{1}{MH}
                \sum_{m=1}^{M}
                \sum_{\ell=0}^{H-1}
                r\!\left(
                    \bfH_{c,i+\ell}^{(m)},
                    \bfQ_{i+\ell}
                \right).
            \]

            \STATE Solve
            \[
                \widehat{\bfJ}_i
                =
                \left(
                    \bfQ_w
                    +
                    \bfF\widehat{\bfJ}_i^{-1}\bfF^T
                \right)^{-1}
                +
                \widehat{\bfM}_i
            \]
            and evaluate
            \[
                \widehat D_i
                =
                \tr{
                    \bfS
                    \widehat{\bfJ}_i^{-1}
                    \bfS^T
                }.
            \]

            \STATE Select the covariance sequence maximizing
            $\widehat R_i-\lambda\widehat D_i$ and denote its first action by
            $\bfQ_i^\star$.

            \STATE Store $(\phi_i,\bfQ_i^\star)$ in $\cD$.

            \STATE Apply $\bfQ_i^\star$ and advance the training trajectory.
        \ENDFOR

        \STATE Standardize the stored predictive-belief features.

        \STATE \textit{Phase 2: Stationary deployment}

        \STATE For the current feature $\phi$, select
        \[
            i^\star
            =
            \arg\min_i
            \left\|
                \widetilde{\phi}
                -
                \widetilde{\phi}_i
            \right\|_2.
        \]

        \STATE Apply
        \[
            \pi_{\rm NN}(\phi)
            =
            \bfQ_{i^\star}^{\star}.
        \]

        \RETURN $\pi_{\rm NN}$.
    \end{algorithmic}
\end{algorithm}

\subsection{Numerical Results and Policy Interpretation}
\label{subsec:numerical_results}

Figure~\ref{fig:rd_curves} shows the resulting trade-off. Moving toward
communication-optimal operation increases the communication rate at the
expense of tracking performance. The PCRB-based outer frontier achieves
smaller distortion than the EKF-based achievable frontier at comparable
rates, consistently with the PCRB being a lower bound on the mean-square
tracking error of any causal estimator.

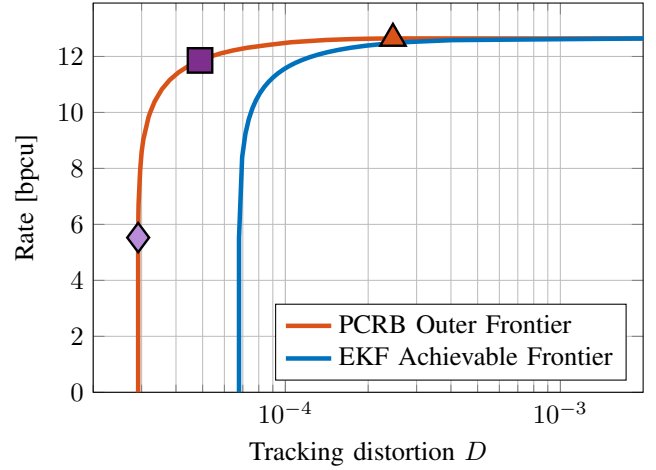
\begin{figure}[t]
	\centering
\begin{tikzpicture}
\definecolor{PCRBOrange}{rgb}{0.8500,0.3250,0.0980}
\definecolor{EKFBlue}{rgb}{0,0.4470,0.7410}
\definecolor{UKFLightPurple}{rgb}{0.72,0.55,0.85}
\definecolor{TradePurple}{rgb}{0.4940,0.1840,0.5560}
\definecolor{CommMarkerOrange}{rgb}{0.8500,0.3250,0.0980}
\definecolor{TradeMarkerPurple}{rgb}{0.4940,0.1840,0.5560}
\definecolor{SenseMarkerLightPurple}{rgb}{0.72,0.55,0.85}
\begin{axis}[
  width=\columnwidth,
  height=0.76\columnwidth,
  xmode=log,
  xmin=2e-05, xmax=0.002,
  ymin=0,
  xlabel={Tracking distortion $D$},
  ylabel={Rate [bpcu]},
  grid=both,
  legend style={at={(0.97,0.03)},anchor=south east},
  legend cell align={left}
]
\addplot+[solid, color=PCRBOrange, line width=1.7pt, mark=none] coordinates {
  (2.914634910923807e-05,0)
  (2.914634910923807e-05,5.528476254343048)
  (2.921649243869362e-05,6.461376118655975)
  (2.925616324553692e-05,6.751711356999857)
  (2.962244905056718e-05,7.904893228333536)
  (3.001027975654091e-05,8.597525825228617)
  (3.044173300257489e-05,9.035617611411068)
  (3.179500338408258e-05,9.833242907833213)
  (3.223634660598023e-05,10.0027847484811)
  (3.346207381050148e-05,10.39060860650423)
  (3.554110206961132e-05,10.82722066318269)
  (3.798090117823132e-05,11.16186984804259)
  (4.089524678644614e-05,11.4332943705804)
  (4.445431727360625e-05,11.66163604356949)
  (4.892601505971811e-05,11.85871920366119)
  (5.476107956000613e-05,12.03208202083541)
  (5.824848061813392e-05,12.1110812396278)
  (6.44654639732911e-05,12.21474655242257)
  (7.247413318724827e-05,12.30932979459188)
  (8.685309802179315e-05,12.41272959505217)
  (0.0001047778352709261,12.50792383268438)
  (0.0001403844148795291,12.58855301891954)
  (0.0001783306830134371,12.6242933667889)
  (0.0002043307951303561,12.63530472947975)
  (0.0002270302216627503,12.640838500864)
  (0.0002461120976056938,12.64408159311742)
  (0.002212388626998192,12.64408159311742)
};
\addlegendentry{PCRB Outer Frontier}
\addplot+[solid, color=EKFBlue, line width=1.7pt, mark=none] coordinates {
  (6.781873003391879e-05,0)
  (6.781873003391879e-05,5.51557905849651)
  (6.96155325290873e-05,8.40104059120023)
  (7.154197982021259e-05,9.236971692021729)
  (7.361377375158499e-05,9.760917280692651)
  (7.58505682785232e-05,10.1439390298343)
  (7.826613826226591e-05,10.4461353684383)
  (8.08953693793762e-05,10.6957605141245)
  (8.37569607134999e-05,10.9085221335086)
  (8.69124161726493e-05,11.0938549642311)
  (9.03793844265479e-05,11.258111970752)
  (9.42175362588225e-05,11.4055990011828)
  (9.85316095831997e-05,11.5392952512909)
  (0.000103382607113266,11.6616277592464)
  (0.00010894488503504,11.7743757947138)
  (0.000115354056101658,11.8789939386044)
  (0.000122826900022709,11.9765608508487)
  (0.000131675573441084,12.0679282459686)
  (0.000142430250538815,12.1538370635089)
  (0.000155824755829199,12.2349034957142)
  (0.000173042905792263,12.3116733297284)
  (0.000196346167084156,12.3845675412321)
  (0.000230067540229621,12.4539514112725)
  (0.000285563582736034,12.5201485336583)
  (0.000403425740689459,12.5834387861963)
  (0.00209517596102531,12.6440815931129)
  (0.002212388626998192,12.6440815931129)
};
\addlegendentry{EKF Achievable Frontier}
\addplot[only marks, mark=triangle*, mark size=5.5pt, color=black, mark options={solid,draw=black,fill=CommMarkerOrange,line width=0.9pt}, forget plot] coordinates {(0.0002461120976056938,12.64408159311742)};
\addplot[only marks, mark=square*, mark size=4.6pt, color=black, mark options={solid,draw=black,fill=TradeMarkerPurple,line width=0.9pt}, forget plot] coordinates {(4.892601505971811e-05,11.85871920366119)};
\addplot[only marks, mark=diamond*, mark size=5.5pt, color=black, mark options={solid,draw=black,fill=SenseMarkerLightPurple,line width=0.9pt}, forget plot] coordinates {(2.914634910923807e-05,5.528476254343048)};
\end{axis}
\end{tikzpicture}
	\caption{EKF-based achievable frontier and PCRB-based Gaussian
    covariance-control outer frontier for the deterministic rank-one communication
    channel.}
	\label{fig:rd_curves}
\end{figure}

To illustrate the spatial behavior associated with the rate--distortion
operating points, Figs.~\ref{fig:PCRB_easy_rad_pattern} and
\ref{fig:rad_patterns_ekf_easy} show representative transmit beampatterns for
communication-optimal, trade-off, and tracking-optimal operation in the
baseline tracking scenario. As the operating point moves toward
tracking-optimal operation, progressively more energy is directed toward the
target trajectory, while the trade-off operating point maintains simultaneous
illumination of the sensing and communication directions. In this scenario,
the EKF remains well aligned with the target, and its qualitative spatial
behavior is close to that of the PCRB-based policy. Within each three-panel
comparison, the same underlying target trajectory is retained across the three
operating points to isolate the effect of the transmit policy.

\begin{figure}[t]
    \centering
    \input{Figures/PCRB_easy}
    \caption{Normalized transmit beampatterns over time for three
    PCRB-frontier operating points in Fig.~\ref{fig:rd_curves}, from top to
    bottom: communication-optimal, trade-off, and tracking-optimal operation
    in the baseline tracking scenario. Colors are clipped to $[-30,0]$ dB; black
    dashed: communication direction; magenta solid: true target angle.}
	\label{fig:PCRB_easy_rad_pattern}
    
    \vspace{\floatsep}
    \input{Figures/EKF_easy}
    \caption{Normalized transmit beampatterns over time for three
    EKF-achievable operating points in Fig.~\ref{fig:rd_curves}, from top to
    bottom: communication-optimal, trade-off, and tracking-optimal operation
    in the baseline tracking scenario. Colors are clipped to $[-30,0]$ dB; black
    dashed: communication direction; magenta solid: true target angle; white solid: EKF estimate.}
    \label{fig:rad_patterns_ekf_easy}
\end{figure}

We next repeat the same spatial comparison in a high-mobility tracking
scenario. The target starts at $\theta_0=-72^\circ$ and $d_0=20$ m, with
$\dot\theta_0=15^\circ/{\rm s}$ and $\dot d_0=14$ m/s. The block sensing SNR is
$16$ dB, and the angular components of the initial uncertainty and process
noise are increased relative to the baseline tracking scenario; all other
parameters follow Table~\ref{tab:sim_params}.

Figures~\ref{fig:rad_patterns_pcrb_difficult} and
\ref{fig:rad_patterns_loss} show the corresponding PCRB- and EKF-based
beampatterns. The same target realization is used for both policies and is
retained across the three operating points. The two policies exhibit nearly
identical spatial behavior, including the loss of target illumination and the
subsequent steering of sensing energy along the same incorrect angular path.
The common target realization does not imply that the predictive beliefs or
the covariance matrices selected by the two policies are identical. Rather,
the figures show that differences between the PCRB and EKF tracking measures
can arise without qualitatively different spatial illumination.

\begin{figure}[t]
    \centering
    \input{Figures/PCRB_difficult}
    \caption{Normalized transmit beampatterns over time for three PCRB-based
    operating points, from top to bottom: communication-optimal, trade-off,
    and tracking-optimal operation in the high-mobility tracking scenario.
    Colors are clipped to $[-30,0]$ dB; black dashed: communication direction;
    magenta solid: true target angle.}
    \label{fig:rad_patterns_pcrb_difficult}

    \vspace{\floatsep}
    \input{Figures/EKF_difficult}
    \caption{Normalized transmit beampatterns over time for three EKF-based
    operating points, from top to bottom: communication-optimal, trade-off,
    and tracking-optimal operation in the high-mobility tracking scenario.
    Colors are clipped to $[-30,0]$ dB; black dashed: communication direction;
    magenta solid: true target angle; white solid: EKF estimate.}
    \label{fig:rad_patterns_loss}
\end{figure}

\section{Discussion}

The beampattern comparison provides additional insight into the difference
between the PCRB characterization and the achievable EKF performance. The
achievable distortion in~\eqref{eq:achievable_distortion_sim} is the realized
tracking error with respect to the true target state, whereas the PCRB
distortion in~\eqref{eq:pcrb_distortion_sim} is obtained from the
policy-averaged local data information evaluated over the predictive belief,
followed by the stationary PCRB information equation. Thus, the two tracking
measures can differ even when the corresponding policies exhibit similar
spatial illumination.

More advanced tracking algorithms can improve the achievable frontier by
reducing the realized tracking error and by providing more accurate state
information for adaptive covariance selection. In contrast, the theoretical
PCRB outer characterization is defined by the Bayesian predictive law induced
by the system model and transmission policy and is independent of the
particular estimator used by an achievable scheme. Investigating nonlinear
estimators such as the unscented Kalman filter, particle filters, and
learning-based tracking methods is therefore a natural direction for tightening
the achievable frontier.

\section{Conclusion}

We studied a dynamic ISAC problem in which a multi-antenna transmitter
communicates while tracking a moving target. Sensing performance was
formulated through a long-run MMSE distortion in angle and distance, with the
PCRB used to construct a tractable outer relaxation, while
communication was measured using the Verd\'u--Han inf-information rate to
accommodate general causal and possibly non-stationary policies. By
introducing a belief-augmented information state, we showed that stationary
randomized Markov policies suffice for the covariance-based scalarized problem
used in the outer bound under the stated regularity assumptions. 
We then derived a Gaussian
covariance-control outer bound and showed that Gaussian signaling is
asymptotically without loss of optimality in the large-block regime.
Numerical results illustrated the resulting rate--tracking trade-off by
comparing an EKF-based achievable frontier with a PCRB-based Gaussian
covariance-control outer frontier.

\appendices

\section{Technical Lemmas for the Stationary-Policy Result}
\label{app:stationary_lemmas}

This appendix collects the technical lemmas used in the proof of
Proposition~\ref{prop:stationary_suffices_conf}. Define
\begin{equation}
    \mathcal A(\bfJ)
    \triangleq
    \left(
        \bfQ_w+\bfF\bfJ^{-1}\bfF^T
    \right)^{-1},
    \qquad
    \bfJ\succ\bf0,
    \label{eq:app_A_def}
\end{equation}
and define the scalar distortion measure
\begin{equation}
    d(\bfJ)
    \triangleq
    \tr{
        \bfS\bfJ^{-1}\bfS^T
    }.
    \label{eq:app_d_def}
\end{equation}
For a state--action pair
$\bfz=(\Pi,\bfH_c)$ and $a$, we also write
\begin{equation}
    \bfG(\bfz,a)
    \triangleq
    \bfJ_D(\Pi,\bfQ(a)).
    \label{eq:app_G_def}
\end{equation}
Under Assumption~\ref{ass:mdp_regular}, there exists a finite constant
$\bar g$ such that
\begin{equation}
    {\bf0}
    \preceq
    \bfG(\bfz,a)
    \preceq
    \bar g\bfI
    \label{eq:app_G_bounded}
\end{equation}
on the admissible state--action graph.

\begin{lemma}[Properties of the prediction-information map]
\label{lem:app_A_properties}
On the positive-definite cone, the map $\mathcal A$ in
\eqref{eq:app_A_def} satisfies:
\begin{enumerate}
    \item if $\bfJ_1\preceq\bfJ_2$, then
    \begin{equation}
        \mathcal A(\bfJ_1)
        \preceq
        \mathcal A(\bfJ_2);
    \end{equation}

    \item for every $\alpha\in[0,1]$,
    \begin{align}
        &\mathcal A
        \left(
            \alpha\bfJ_1+(1-\alpha)\bfJ_2
        \right)
        \notag\\
        &\qquad\succeq
        \alpha\mathcal A(\bfJ_1)
        +(1-\alpha)\mathcal A(\bfJ_2);
        \label{eq:app_A_concavity}
    \end{align}

    \item for every $\bfJ\succ\bf0$,
    \begin{equation}
        \bf0
        \prec
        \mathcal A(\bfJ)
        \preceq
        \bfQ_w^{-1}.
        \label{eq:app_A_bound}
    \end{equation}
\end{enumerate}
\end{lemma}

\begin{IEEEproof}
If $\bfJ_1\preceq\bfJ_2$, then
$\bfJ_1^{-1}\succeq\bfJ_2^{-1}$. Hence,
\begin{equation}
    \bfQ_w+\bfF\bfJ_1^{-1}\bfF^T
    \succeq
    \bfQ_w+\bfF\bfJ_2^{-1}\bfF^T,
\end{equation}
and inversion reverses the Loewner order, which proves monotonicity.

To establish concavity, the Woodbury identity gives
\begin{align}
    \mathcal A(\bfJ)
    &=
    \bfQ_w^{-1}
    -
    \bfQ_w^{-1}\bfF
    \left(
        \bfJ+\bfF^T\bfQ_w^{-1}\bfF
    \right)^{-1}
    \bfF^T\bfQ_w^{-1}.
    \label{eq:app_A_woodbury}
\end{align}
The inverse map is operator convex on the positive-definite cone. Therefore,
\begin{align}
    &
    \left(
        \alpha\bfJ_1
        +(1-\alpha)\bfJ_2
        +\bfF^T\bfQ_w^{-1}\bfF
    \right)^{-1}
    \notag\\
    &\quad\preceq
    \alpha
    \left(
        \bfJ_1+\bfF^T\bfQ_w^{-1}\bfF
    \right)^{-1}
    +(1-\alpha)
    \left(
        \bfJ_2+\bfF^T\bfQ_w^{-1}\bfF
    \right)^{-1}.
\end{align}
Substitution into~\eqref{eq:app_A_woodbury} yields
\eqref{eq:app_A_concavity}. Finally,
\begin{equation}
    \bfQ_w+\bfF\bfJ^{-1}\bfF^T
    \succeq
    \bfQ_w\succ\bf0,
\end{equation}
which gives~\eqref{eq:app_A_bound}.
\end{IEEEproof}

\begin{lemma}[Uniform bounds on the PCRB information sequence]
\label{lem:app_J_compact}
Under Assumption~\ref{ass:mdp_regular}, there exist constants
$0<\underline j\leq\overline j<\infty$, independent of the admissible policy
and of $k$, such that
\begin{equation}
    \underline j\bfI
    \preceq
    \bfJ_k^\pi
    \preceq
    \overline j\bfI,
    \qquad
    k\geq0.
    \label{eq:app_J_compact}
\end{equation}
\end{lemma}

\begin{IEEEproof}
From~\eqref{eq:PCRB_recursion},
Lemma~\ref{lem:app_A_properties}, and
\eqref{eq:app_G_bounded},
\begin{equation}
    \bfJ_k^\pi
    =
    \mathcal A(\bfJ_{k-1}^\pi)
    +
    \bfM_k^\pi
    \preceq
    \bfQ_w^{-1}
    +
    \bar g\bfI.
\end{equation}
Hence a policy-independent finite upper bound exists.

For the lower bound, define the no-measurement information recursion
\begin{equation}
    \bfK_0=\bfJ_0,
    \qquad
    \bfK_k=\mathcal A(\bfK_{k-1}).
    \label{eq:app_K_recursion}
\end{equation}
Since $\bfM_k^\pi\succeq\bf0$ and $\mathcal A$ is monotone,
induction gives
\begin{equation}
    \bfJ_k^\pi\succeq\bfK_k
\end{equation}
for every $k$. Let $\bfP_k\triangleq\bfK_k^{-1}$. Then
\begin{equation}
    \bfP_k
    =
    \bfQ_w
    +
    \bfF\bfP_{k-1}\bfF^T,
\end{equation}
and therefore
\begin{equation}
    \bfP_k
    =
    \bfF^k\bfP_0(\bfF^T)^k
    +
    \sum_{i=0}^{k-1}
    \bfF^i\bfQ_w(\bfF^T)^i.
    \label{eq:app_no_measurement_covariance}
\end{equation}
Here $\bfP_0=\bfJ_0^{-1}$. Because $\bfF$ is Schur stable, the first term vanishes and the matrix series
is bounded. Hence there exists $c<\infty$ such that
$\bfP_k\preceq c\bfI$ for all $k$. Thus
\begin{equation}
    \bfK_k
    =
    \bfP_k^{-1}
    \succeq
    c^{-1}\bfI,
\end{equation}
and consequently
$\bfJ_k^\pi\succeq c^{-1}\bfI$ uniformly in $k$ and $\pi$.
\end{IEEEproof}

\begin{lemma}[Convexity and order monotonicity of the PCRB distortion]
\label{lem:app_d_properties}
The function $d$ in~\eqref{eq:app_d_def} is convex on the
positive-definite cone and decreasing in the Loewner order. In particular,
for any positive-definite
$\bfJ_1,\ldots,\bfJ_N$,
\begin{equation}
    \frac{1}{N}
    \sum_{k=1}^{N}
    d(\bfJ_k)
    \geq
    d\!\left(
        \frac{1}{N}
        \sum_{k=1}^{N}
        \bfJ_k
    \right).
    \label{eq:app_d_Jensen}
\end{equation}
\end{lemma}

\begin{IEEEproof}
Operator convexity of matrix inversion gives
\begin{equation}
    \left(
        \frac{1}{N}
        \sum_{k=1}^{N}
        \bfJ_k
    \right)^{-1}
    \preceq
    \frac{1}{N}
    \sum_{k=1}^{N}
    \bfJ_k^{-1}.
\end{equation}
Congruence by $\bfS$ and the monotonicity of the trace on the
positive-semidefinite cone yield~\eqref{eq:app_d_Jensen}.

Moreover, if $\bfJ_1\preceq\bfJ_2$, then
$\bfJ_1^{-1}\succeq\bfJ_2^{-1}$, and hence
\begin{equation}
    d(\bfJ_1)
    \geq
    d(\bfJ_2).
\end{equation}
\end{IEEEproof}

\begin{lemma}[Weak limits of occupation measures]
\label{lem:app_occupation}
Let
\begin{equation}
    \mu_{\pi,N}(B)
    =
    \frac{1}{N}
    \sum_{k=1}^{N}
    \bbP^\pi
    \left(
        (\bfZ_k,a_k)\in B
    \right)
    \label{eq:app_empirical_occupation_measure}
\end{equation}
be the expected empirical state--action occupation measure induced by an
admissible causal policy $\pi$. Under
Assumption~\ref{ass:mdp_regular}, if
\begin{equation}
    \mu_{\pi,N_j}
    \Rightarrow
    \mu,
\end{equation}
then $\mu$ is an invariant occupation measure. In particular, if $\nu$ is its
state marginal, there exists a stationary randomized Markov kernel
$\varphi(da|\bfz)$ such that
\begin{equation}
    \mu(d\bfz,da)
    =
    \nu(d\bfz)
    \varphi(da|\bfz)
    \label{eq:app_disintegration}
\end{equation}
and
\begin{equation}
    \nu(B)
    =
    \int
    P(B|\bfz,a)
    \,\mu(d\bfz,da)
    \label{eq:app_flow_constraint}
\end{equation}
for every Borel set $B$.
\end{lemma}

\begin{IEEEproof}
Let $f$ be bounded and continuous on the information-state space, and define
\begin{equation}
    Pf(\bfz,a)
    \triangleq
    \int
    f(\bfz')
    P(d\bfz'|\bfz,a).
\end{equation}
Weak continuity of the controlled transition kernel implies that
$Pf$ is bounded and continuous on the admissible graph. From the controlled
Markov property,
\begin{align}
    &
    \int
    f(\bfz)
    \,\mu_{\pi,N}^{Z}(d\bfz)
    -
    \int
    Pf(\bfz,a)
    \,\mu_{\pi,N}(d\bfz,da)
    \notag\\
    &=
    \frac{
        \E^\pi[f(\bfZ_1)]
        -
        \E^\pi[f(\bfZ_{N+1})]
    }{N},
    \label{eq:app_flow_telescoping}
\end{align}
where $\mu_{\pi,N}^{Z}$ denotes the state marginal of
$\mu_{\pi,N}$. Since $f$ is bounded, the right-hand side converges to zero.
Passing to the weakly convergent subsequence gives
\eqref{eq:app_flow_constraint}.

Since the state and action spaces are Polish, the disintegration theorem
provides a stochastic kernel $\varphi(da|\bfz)$ satisfying
\eqref{eq:app_disintegration}. If $\bfZ_1\sim\nu$ and the actions are selected
according to $\varphi$, then~\eqref{eq:app_flow_constraint} implies
\begin{equation}
    (\bfZ_k,a_k)\sim\mu
\end{equation}
for every $k$.
\end{IEEEproof}

\begin{lemma}[Constant-information PCRB fixed point]
\label{lem:app_riccati_fixed}
Fix any $\bfM\succeq\bf0$. Under
Assumption~\ref{ass:mdp_regular}, the recursion
\begin{equation}
    \bfJ_{n+1}
    =
    \mathcal A(\bfJ_n)
    +
    \bfM,
    \qquad
    \bfJ_0\succ\bf0,
    \label{eq:app_constant_M_recursion}
\end{equation}
has a unique positive-definite fixed point
$\bfJ(\bfM)$ and
\begin{equation}
    \bfJ_n
    \longrightarrow
    \bfJ(\bfM)
    \label{eq:app_J_convergence}
\end{equation}
for every $\bfJ_0\succ\bf0$.
\end{lemma}

\begin{IEEEproof}
Let $\bfP_n\triangleq\bfJ_n^{-1}$ and choose a matrix $\bfC$ satisfying
$\bfC^T\bfC=\bfM$, e.g., $\bfC=\bfM^{1/2}$. Then
\eqref{eq:app_constant_M_recursion} is equivalent to
\begin{equation}
    \bfP_{n+1}
    =
    \left[
        \left(
            \bfQ_w+\bfF\bfP_n\bfF^T
        \right)^{-1}
        +
        \bfC^T\bfC
    \right]^{-1}.
    \label{eq:app_Kalman_Riccati}
\end{equation}
This is the covariance Riccati recursion of a time-invariant Kalman filter
with state matrix $\bfF$, process-noise covariance $\bfQ_w$, measurement
matrix $\bfC$, and identity measurement-noise covariance.

Since $\bfQ_w\succ\bf0$, the pair
$(\bfF,\bfQ_w^{1/2})$ is stabilizable. Moreover, because $\bfF$ is Schur
stable, $(\bfF,\bfC)$ is detectable for every $\bfC$. The standard
discrete-time Riccati convergence theorem therefore implies that
\eqref{eq:app_Kalman_Riccati} converges, from every positive-definite initial
condition, to its unique stabilizing positive-definite solution
$\bfP(\bfM)$. Hence $\bfJ(\bfM)\triangleq\bfP(\bfM)^{-1}$ is the unique
positive-definite solution of
\begin{equation}
    \bfJ(\bfM)
    =
    \mathcal A(\bfJ(\bfM))
    +
    \bfM,
    \label{eq:app_J_fixed_point}
\end{equation}
and~\eqref{eq:app_J_convergence} follows.
\end{IEEEproof}

\section{Proof of Proposition~\ref{prop:stationary_suffices_conf}}
\label{app:proof_stationary_suffices}

We prove Proposition~\ref{prop:stationary_suffices_conf}. Fix an arbitrary
admissible causal policy $\pi$ and define the finite-horizon averages
\begin{align}
    R_N^\pi
    &\triangleq
    \frac{1}{N}
    \sum_{k=1}^{N}
    \E^\pi
    \left[
        r_{\rm G}(\bfZ_k,a_k)
    \right],
    \label{eq:app_RN}\\
    D_N^\pi
    &\triangleq
    \frac{1}{N}
    \sum_{k=1}^{N}
    d(\bfJ_k^\pi),
    \label{eq:app_DN}\\
    \bar{\bfM}_N^\pi
    &\triangleq
    \frac{1}{N}
    \sum_{k=1}^{N}
    \bfM_k^\pi.
    \label{eq:app_MbarN}
\end{align}
By~\eqref{eq:gaussian_average_reward} and
\eqref{eq:tracking_distortion_conf},
\begin{equation}
    R_{\rm G}(\pi)
    =
    \liminf_{N\to\infty}
    R_N^\pi,
    \qquad
    \overline{D}(\pi)
    =
    \limsup_{N\to\infty}
    D_N^\pi.
    \label{eq:app_limit_definitions}
\end{equation}

Let $\mu_{\pi,N}$ denote the expected empirical occupation measure in
\eqref{eq:app_empirical_occupation_measure}. Since
$\bfM_k^\pi
=
\E^\pi[\bfG(\bfZ_k,a_k)]$, we have
\begin{align}
    R_N^\pi
    &=
    \int
        r_{\rm G}(\bfz,a)
        \,\mu_{\pi,N}(d\bfz,da),
    \label{eq:app_R_as_occupation}\\
    \bar{\bfM}_N^\pi
    &=
    \int
        \bfG(\bfz,a)
        \,\mu_{\pi,N}(d\bfz,da).
    \label{eq:app_M_as_occupation}
\end{align}

\subsection*{A. Limiting Occupation Measure}

By the definition of the superior limit, there exists a sequence
$N_j\uparrow\infty$ such that
\begin{equation}
    D_{N_j}^\pi
    \longrightarrow
    \overline{D}(\pi).
    \label{eq:app_D_subsequence}
\end{equation}
By Assumption~\ref{ass:mdp_regular}, the sequence
$\{\mu_{\pi,N_j}\}$ is tight. Hence, by Prokhorov's theorem, there exists a
further subsequence, not relabeled, and a probability measure $\mu$ such that
\begin{equation}
    \mu_{\pi,N_j}
    \Rightarrow
    \mu.
    \label{eq:app_mu_convergence}
\end{equation}
The convergence in~\eqref{eq:app_D_subsequence} is preserved along this
further subsequence.

By Lemma~\ref{lem:app_occupation}, $\mu$ is an invariant occupation measure
and admits the disintegration
\begin{equation}
    \mu(d\bfz,da)
    =
    \nu(d\bfz)\,
    \varphi(da|\bfz),
    \label{eq:app_mu_disintegration_proof}
\end{equation}
where $\varphi$ is a stationary randomized Markov policy and $\nu$ is an
invariant state distribution.

By Assumption~\ref{ass:mdp_regular}, weak convergence in
\eqref{eq:app_mu_convergence} preserves the communication-reward and
data-information averages, yielding
\begin{align}
    R_{N_j}^\pi
    &\longrightarrow
    R_{\mu}
    \triangleq
    \int
        r_{\rm G}(\bfz,a)
        \,\mu(d\bfz,da),
    \label{eq:app_R_mu}\\
    \bar{\bfM}_{N_j}^\pi
    &\longrightarrow
    \bfM_{\mu}
    \triangleq
    \int
        \bfG(\bfz,a)
        \,\mu(d\bfz,da).
    \label{eq:app_M_mu}
\end{align}
Since every subsequential limit of $\{R_N^\pi\}$ is no smaller than its
inferior limit,
\begin{equation}
    R_{\mu}
    \geq
    R_{\rm G}(\pi).
    \label{eq:app_rate_dominance_intermediate}
\end{equation}

\subsection*{B. Averaging the Deterministic PCRB Recursion}

Define
\begin{align}
    \bar{\bfJ}_N^\pi
    &\triangleq
    \frac{1}{N}\sum_{k=1}^{N}\bfJ_k^\pi,
    &
    \widetilde{\bfJ}_N^\pi
    &\triangleq
    \frac{1}{N}\sum_{k=0}^{N-1}\bfJ_k^\pi.
\end{align}
Averaging~\eqref{eq:PCRB_recursion} and applying the matrix concavity of
$\mathcal A$ from Lemma~\ref{lem:app_A_properties} give
\begin{align}
    \bar{\bfJ}_N^\pi
    &=
    \frac{1}{N}
    \sum_{k=1}^{N}
    \mathcal A(\bfJ_{k-1}^\pi)
    +
    \bar{\bfM}_N^\pi
    \notag\\
    &\preceq
    \mathcal A\!\left(
        \widetilde{\bfJ}_N^\pi
    \right)
    +
    \bar{\bfM}_N^\pi.
    \label{eq:app_average_recursion_inequality}
\end{align}
Moreover,
\begin{equation}
    \widetilde{\bfJ}_N^\pi
    -
    \bar{\bfJ}_N^\pi
    =
    \frac{\bfJ_0-\bfJ_N^\pi}{N}.
    \label{eq:app_boundary_term}
\end{equation}
By Lemma~\ref{lem:app_J_compact}, the right-hand side converges to zero.

Again by Lemma~\ref{lem:app_J_compact}, after passing to a further
subsequence if necessary,
\begin{equation}
    \bar{\bfJ}_{N_j}^\pi
    \longrightarrow
    \widehat{\bfJ}
    \succ\bf0.
    \label{eq:app_Jhat}
\end{equation}
The boundary relation then also gives
$\widetilde{\bfJ}_{N_j}^\pi\to\widehat{\bfJ}$. Passing to the limit in
\eqref{eq:app_average_recursion_inequality}, using continuity of
$\mathcal A$ and~\eqref{eq:app_M_mu}, yields
\begin{equation}
    \widehat{\bfJ}
    \preceq
    \mathcal A(\widehat{\bfJ})
    +
    \bfM_\mu.
    \label{eq:app_J_subsolution}
\end{equation}

\subsection*{C. Comparison With the Constant-Information Fixed Point}

Define
$\mathcal T_\mu(\bfJ)\triangleq\mathcal A(\bfJ)+\bfM_\mu$.
By~\eqref{eq:app_J_subsolution},
$\widehat{\bfJ}\preceq\mathcal T_\mu(\widehat{\bfJ})$.
Since $\mathcal T_\mu$ is monotone by
Lemma~\ref{lem:app_A_properties},
\begin{equation}
    \widehat{\bfJ}
    \preceq
    \mathcal T_\mu(\widehat{\bfJ})
    \preceq
    \mathcal T_\mu^2(\widehat{\bfJ})
    \preceq
    \cdots.
    \label{eq:app_monotone_iteration}
\end{equation}
Moreover,
$\mathcal T_\mu(\bfJ)\preceq\bfQ_w^{-1}+\bfM_\mu$
for every $\bfJ\succ\bf0$ by~\eqref{eq:app_A_bound}. Hence the monotone
sequence in~\eqref{eq:app_monotone_iteration} is bounded above and converges
to a positive-definite fixed point of $\mathcal T_\mu$. By
Lemma~\ref{lem:app_riccati_fixed}, this fixed point is uniquely
$\bfJ(\bfM_\mu)$. Consequently,
\begin{equation}
    \widehat{\bfJ}
    \preceq
    \bfJ(\bfM_\mu).
    \label{eq:app_Jhat_fixed_comparison}
\end{equation}

\subsection*{D. Distortion Dominance}

By Lemma~\ref{lem:app_d_properties},
\begin{equation}
    D_N^\pi
    \geq
    d\!\left(
        \bar{\bfJ}_N^\pi
    \right).
    \label{eq:app_distortion_Jensen}
\end{equation}
Taking $N=N_j$ and using
\eqref{eq:app_D_subsequence} and~\eqref{eq:app_Jhat} gives
$\overline{D}(\pi)\geq d(\widehat{\bfJ})$. Since $d$ is decreasing in the Loewner
order,~\eqref{eq:app_Jhat_fixed_comparison} therefore yields
\begin{equation}
    \overline{D}(\pi)
    \geq
    d\!\left(
        \bfJ(\bfM_\mu)
    \right).
    \label{eq:app_D_lower_stationary}
\end{equation}

\subsection*{E. Stationary Policy Induced by the Limiting Occupation Measure}

Consider the stationary randomized Markov policy $\varphi$ obtained from
\eqref{eq:app_mu_disintegration_proof}, initialized according to its invariant
state distribution $\nu$. By Lemma~\ref{lem:app_occupation},
$(\bfZ_k,a_k)\sim\mu$ for every $k$. Hence
\begin{equation}
    \E^\varphi[
        r_{\rm G}(\bfZ_k,a_k)
    ]
    =
    R_\mu,
    \qquad
    \bfM_k^\varphi
    =
    \bfM_\mu,
    \qquad
    k\geq1.
\end{equation}
Therefore
$R_{\rm G}(\varphi)=R_\mu$, while its deterministic PCRB recursion is
\begin{equation}
    \bfJ_k^\varphi
    =
    \mathcal A(\bfJ_{k-1}^\varphi)
    +
    \bfM_\mu.
\end{equation}
Lemma~\ref{lem:app_riccati_fixed} gives
$\bfJ_k^\varphi\to\bfJ(\bfM_\mu)$, and continuity of $d$ together with
Ces\`aro convergence yields
\begin{equation}
    \overline{D}(\varphi)
    =
    d\!\left(
        \bfJ(\bfM_\mu)
    \right).
    \label{eq:app_D_phi}
\end{equation}
By Assumption~\ref{ass:mdp_regular}, these same long-run rate and
distortion values hold when $\varphi$ is initialized according to the
prescribed initial information-state distribution.

Combining this identity with
\eqref{eq:app_rate_dominance_intermediate} and
\eqref{eq:app_D_lower_stationary} gives
\begin{equation}
    R_{\rm G}(\varphi)\geq R_{\rm G}(\pi),
    \qquad
    \overline{D}(\varphi)\leq \overline{D}(\pi),
\end{equation}
which proves~\eqref{eq:stationary_dominance_rate} and
\eqref{eq:stationary_dominance_distortion}.

Since stationary randomized Markov policies are admissible causal policies,
one direction of~\eqref{eq:stationary_outer_problem} is immediate, while the
reverse direction follows from the dominance result above. Hence
\eqref{eq:stationary_outer_problem} holds.

\subsection*{F. Attainment}

It remains to establish that the stationary supremum is attained. Let
$\cK$ denote the weakly compact set of invariant state--action occupation
measures guaranteed by Assumption~\ref{ass:mdp_regular}. For
$\mu\in\cK$, define
\begin{equation}
    \bfM(\mu)
    \triangleq
    \int
        \bfG(\bfz,a)
        \,\mu(d\bfz,da)
    \label{eq:app_M_of_mu}
\end{equation}
and
\begin{equation}
    \mathcal L_\lambda(\mu)
    \triangleq
    \int
        r_{\rm G}(\bfz,a)
        \,\mu(d\bfz,da)
    -
    \lambda
    d
    \left(
        \bfJ(\bfM(\mu))
    \right).
    \label{eq:app_stationary_functional}
\end{equation}

We first establish continuity of the map
$\bfM\mapsto\bfJ(\bfM)$ on the set of average data-information matrices
induced by $\cK$. Let $\bfM_n\to\bfM$ and set
$\bfJ_n\triangleq\bfJ(\bfM_n)$. Since
\[
    \bfJ_n
    =
    \mathcal A(\bfJ_n)+\bfM_n
    \preceq
    \bfQ_w^{-1}+\bfM_n,
\]
and the matrices $\bfM_n$ are uniformly bounded by
Assumption~\ref{ass:mdp_regular}, the bounds of
Lemma~\ref{lem:app_J_compact} place $\{\bfJ_n\}$ in a compact subset of the
positive-definite cone.

For any convergent subsequence
$\bfJ_{n_j}\to\bfJ_\star$, continuity of $\mathcal A$ gives
\[
    \bfJ_\star
    =
    \mathcal A(\bfJ_\star)+\bfM.
\]
Uniqueness in Lemma~\ref{lem:app_riccati_fixed} therefore implies
$\bfJ_\star=\bfJ(\bfM)$. Since every convergent subsequence has the same
limit,
\begin{equation}
    \bfJ(\bfM_n)
    \longrightarrow
    \bfJ(\bfM).
    \label{eq:app_fixed_point_continuity}
\end{equation}

By Assumption~\ref{ass:mdp_regular},
\eqref{eq:app_fixed_point_continuity}, and continuity of $d$, the functional
$\mathcal L_\lambda$ is continuous on the weakly compact set $\cK$.
Weierstrass' theorem therefore gives
$\mu^\star\in\cK$ attaining its maximum. Disintegrate
\[
    \mu^\star(d\bfz,da)
    =
    \nu^\star(d\bfz)\varphi^\star(da|\bfz).
\]
The resulting stationary randomized Markov policy $\varphi^\star$,
initialized with $\nu^\star$, realizes $\mu^\star$ and satisfies
\[
    R_{\rm G}(\varphi^\star)
    -
    \lambda \overline{D}(\varphi^\star)
    =
    \mathcal L_\lambda(\mu^\star).
\]
By Assumption~\ref{ass:mdp_regular}, the same objective value is attained
under the prescribed initial information-state distribution.
Hence $\varphi^\star$ attains the stationary supremum in
\eqref{eq:stationary_outer_problem}, completing the proof.

\section{Proof of the Large-Block Gaussian Reduction}
\label{app:large_block_gaussian}

This appendix first proves~\eqref{eq:covariance_genie_bound} and then gives
the technical arguments underlying Propositions~\ref{prop:large_block_gaussian}
and~\ref{prop:gaussian_outer_bound_conf}.

\paragraph*{Proof of~\eqref{eq:covariance_genie_bound}}
Since
\[
    \bfQ_{1:N}
    \longrightarrow
    (\bfX_{1:N},\bfH_{c,1:N})
    \longrightarrow
    \bfY_{c,1:N},
\]
the conditional information-density chain rule gives
\begin{align}
    i(\bfX_{1:N};\bfY_{c,1:N}&\mid\bfH_{c,1:N})
    =
    i(\bfQ_{1:N};\bfY_{c,1:N}\mid\bfH_{c,1:N})
    \notag\\
    &\quad+
    i(\bfX_{1:N};\bfY_{c,1:N}
      \mid\bfH_{c,1:N},\bfQ_{1:N}).
\end{align}
The standard inequalities for limits in probability therefore imply
\begin{align}
    \underline R(\pi)
    &\leq
    \overline I_{Q,L}(\pi) \nonumber\\
    &+
    \rmp\text{-}\liminf_{N\to\infty}
    \frac{1}{NL}
    i(\bfX_{1:N};\bfY_{c,1:N}
      \mid\bfH_{c,1:N},\bfQ_{1:N}).
\end{align}
The standard lower-tail bound for information density
\cite{verdu_general,han_information_spectrum} gives
\begin{align}
    &\rmp\text{-}\liminf_{N\to\infty}
    \frac{1}{NL}
    i(\bfX_{1:N};\bfY_{c,1:N}
      \mid\bfH_{c,1:N},\bfQ_{1:N}) \nonumber\\
    &\leq
    \liminf_{N\to\infty}
    \frac{1}{NL}
    I(\bfX_{1:N};\bfY_{c,1:N}
      \mid\bfH_{c,1:N},\bfQ_{1:N}).
\end{align}
For the memoryless Gaussian communication channel, the tower property,
Gaussian maximum entropy, and concavity of $\log\det$ yield
\begin{align}
    &\frac{1}{NL}
    I(\bfX_{1:N};\bfY_{c,1:N}
      \mid\bfH_{c,1:N},\bfQ_{1:N})\nonumber\\
    &\leq
    \frac{1}{N}
    \sum_{k=1}^{N}
    \E^\pi
    \left[
        r_{\rm G}(\bfZ_k,a_k)
    \right].
\end{align}
Taking the limit inferior proves~\eqref{eq:covariance_genie_bound}.

The exact finite-block reduction of the sensing likelihood to the sufficient
statistics
$(\bfY_{s,k}\bfX_k^{\herm},\bfX_k\bfX_k^{\herm})$
was established in Lemma~\ref{lem:block_gram_sufficiency}. We now turn to the
large-block covariance limit.

\begin{lemma}[Large-block covariance sufficiency of the belief transition]
\label{lem:app_large_block_belief}
Consider an admissible family of within-block signaling laws satisfying
\eqref{eq:large_block_empirical_covariance} and the sensing normalization
\eqref{eq:large_block_sensing_normalization}. For any fixed predictive belief
$\Pi$ and selected covariance $\bfQ$, all such signaling laws induce the same
limiting Bayesian belief-transition kernel as $L\to\infty$.
\end{lemma}

\begin{IEEEproof}
Suppressing the block index $k$, let
$\bfT_L=\bfY_s\bfX^{\herm}$ and
$\bfG_L=\bfX\bfX^{\herm}$ denote the sufficient statistics identified in
Lemma~\ref{lem:block_gram_sufficiency}, and let
$\widehat{\bfQ}_L=L^{-1}\bfG_L$ and
$\overline{\bfT}_L=L^{-1}\bfT_L$ as in
\eqref{eq:normalized_sufficient_statistics}. Under the true target state
$\bfs_0$,
\begin{equation}
    \overline{\bfT}_L
    =
    \bfH_s(\bfs_0)\widehat{\bfQ}_L
    +
    \bfW_L,
    \qquad
    \bfW_L
    \triangleq
    \frac{1}{L}\bfZ_s\bfX^{\herm}.
\end{equation}
Conditional on $\bfX$, each row of $\bfW_L$ is circularly symmetric complex
Gaussian with covariance
\begin{equation}
    \frac{\sigma_{s,L}^{2}}{L}
    \widehat{\bfQ}_L
    \xrightarrow[L\to\infty]{\mathbb P}
    \bar{\sigma}_s^{2}\bfQ,
\end{equation}
where~\eqref{eq:large_block_empirical_covariance} and
\eqref{eq:large_block_sensing_normalization} were used. Together with
\eqref{eq:large_block_empirical_covariance}, this yields
\begin{equation}
    \left(
        \widehat{\bfQ}_L,
        \overline{\bfT}_L
    \right)
    \Rightarrow
    \left(
        \bfQ,
        \bfH_s(\bfs_0)\bfQ+\bfW_{\bfQ}
    \right),
    \label{eq:app_sufficient_statistic_limit}
\end{equation}
where the rows of $\bfW_{\bfQ}$ are independent circularly symmetric complex
Gaussian vectors with covariance $\bar{\sigma}_s^{2}\bfQ$.

Up to terms independent of $\bfs$, the block log-likelihood can be written as
\begin{align}
    \ell_L(\bfs)
    &=
    \frac{L}{\sigma_{s,L}^{2}}
    \Bigg[
        2\Re\!\left\{
            \tr{
                \bfH_s^{\herm}(\bfs)
                \overline{\bfT}_L
            }
        \right\}
        \notag\\
    &\hspace{2.2cm}
        -
        \tr{
            \bfH_s^{\herm}(\bfs)
            \bfH_s(\bfs)
            \widehat{\bfQ}_L
        }
    \Bigg].
\end{align}
By~\eqref{eq:large_block_sensing_normalization}, 
$L/\sigma_{s,L}^{2}\to1/\bar{\sigma}_s^{2}$. 
By Assumption~\ref{ass:large_block_stationary}, the posterior update is
continuous in the predictive belief and the normalized sufficient statistics.
The continuous-mapping theorem applied to
\eqref{eq:app_sufficient_statistic_limit} therefore shows that the posterior
law has a limit determined only by $(\Pi,\bfQ)$. Applying the common
Gauss--Markov prediction kernel to this posterior gives the same conclusion
for the next predictive belief. Hence all admissible signaling laws
implementing the same covariance $\bfQ$ induce the same limiting
belief-transition kernel.
\end{IEEEproof}

\begin{lemma}[Common stationary large-block limit]
\label{lem:app_common_invariant_large_L}
Fix a stationary randomized covariance policy. Let
$\mu_L^{A}$ and $\mu_L^{G}$ denote invariant occupation measures of the
information state and covariance action under, respectively, an arbitrary
admissible within-block signaling implementation and its conditionally
i.i.d. Gaussian implementation using the same covariance policy. Under
Assumption~\ref{ass:large_block_stationary},
\begin{equation}
    \mu_L^{A}
    \Rightarrow
    \mu_\infty,
    \qquad
    \mu_L^{G}
    \Rightarrow
    \mu_\infty,
    \label{eq:app_common_invariant_limit}
\end{equation}
where $\mu_\infty$ is the unique invariant occupation measure of the common
limiting covariance-controlled kernel.
\end{lemma}

\begin{IEEEproof}
Consider an arbitrary subsequence of $\{\mu_L^{A}\}$. Uniform tightness in
Assumption~\ref{ass:large_block_stationary} yields a further subsequence, not
relabeled, such that
$\mu_L^{A}\Rightarrow\mu$ for some probability measure $\mu$.

For every bounded continuous function $f$ of the information state,
invariance gives
\begin{equation}
    \int
        f(\bfz)
        \,\mu_L^{A}(d\bfz,d\bfQ)
    =
    \int
        P_L^{A}f(\bfz,\bfQ)
        \,\mu_L^{A}(d\bfz,d\bfQ),
    \label{eq:app_large_L_invariance}
\end{equation}
where $P_L^{A}$ denotes the controlled information-state kernel. By
Lemma~\ref{lem:app_large_block_belief} and
Assumption~\ref{ass:large_block_stationary}, these kernels converge uniformly
on compact subsets to the common covariance-controlled kernel $P_\infty$.
Uniform tightness therefore permits passage to the limit in
\eqref{eq:app_large_L_invariance}, giving
\begin{equation}
    \int
        f(\bfz)
        \,\mu(d\bfz,d\bfQ)
    =
    \int
        P_\infty f(\bfz,\bfQ)
        \,\mu(d\bfz,d\bfQ).
\end{equation}
Thus, $\mu$ is invariant for the limiting process. Uniqueness in
Assumption~\ref{ass:large_block_stationary} implies
$\mu=\mu_\infty$. Since every convergent subsequence has the same limit,
$\mu_L^{A}\Rightarrow\mu_\infty$. The identical argument applies to the
conditionally Gaussian implementation, proving
\eqref{eq:app_common_invariant_limit}.
\end{IEEEproof}

\begin{IEEEproof}[Proof of Proposition~\ref{prop:large_block_gaussian}]
Fix a stationary randomized covariance policy and compare an arbitrary
admissible signaling implementation $A$ with its conditionally i.i.d.
Gaussian implementation $G$. By
Lemma~\ref{lem:app_common_invariant_large_L}, both implementations converge
to the same invariant state--covariance occupation measure $\mu_\infty$.

Define
\begin{equation}
    \bfM_L^{i}
    \triangleq
    \int
        \bfJ_D(\Pi,\bfQ)
        \,\mu_L^{i}(d\bfz,d\bfQ),
    \qquad
    i\in\{A,G\}.
\end{equation}
By Assumption~\ref{ass:mdp_regular} and
\eqref{eq:app_common_invariant_limit},
\begin{equation}
    \left\|
        \bfM_L^{G}-\bfM_L^{A}
    \right\|
    \longrightarrow0.
\end{equation}
Continuity of the fixed-point map established in
\eqref{eq:app_fixed_point_continuity} then gives
\begin{equation}
    \left\|
        \bfJ(\bfM_L^{G})
        -
        \bfJ(\bfM_L^{A})
    \right\|
    \longrightarrow0.
\end{equation}
Since $d$ is continuous on the uniformly positive-definite set of
Lemma~\ref{lem:app_J_compact},
\begin{equation}
    D_{{\rm PCRB},L}^{G}
    -
    D_{{\rm PCRB},L}^{A}
    \longrightarrow0,
\end{equation}
which proves~\eqref{eq:large_block_pcrb_equality}.

It remains to compare communication performance. Condition on a
communication-channel realization $\bfH_c$ and covariance $\bfQ$. Let
\begin{equation}
    \bfQ_\ell
    \triangleq
    \E\!\left[
        \bfx_\ell\bfx_\ell^{\herm}
        \,\middle|\,
        \bfH_c,\bfQ
    \right],
    \qquad
    \frac{1}{L}
    \sum_{\ell=1}^{L}\bfQ_\ell
    =
    \bfQ.
\end{equation}
For the memoryless complex Gaussian communication channel,
\begin{align}
    \frac{1}{L}
    I(\bfX;\bfY_c\mid\bfH_c,\bfQ)
    &\leq
    \frac{1}{L}
    \sum_{\ell=1}^{L}
    I(
        \bfx_\ell;
        \bfy_{c,\ell}
        \mid
        \bfH_c,\bfQ
    )
    \notag\\
    &\leq
    \frac{1}{L}
    \sum_{\ell=1}^{L}
    \log_2\det\!\left(
        \bfI+
        \sigma_c^{-2}
        \bfH_c
        \bfQ_\ell
        \bfH_c^{\herm}
    \right)
    \notag\\
    &\leq
    \log_2\det\!\left(
        \bfI+
        \sigma_c^{-2}
        \bfH_c
        \bfQ
        \bfH_c^{\herm}
    \right).
    \label{eq:app_block_gaussian_bound}
\end{align}
The first inequality follows from the memoryless channel structure and
subadditivity of output differential entropy. The second follows from the
Gaussian maximum-entropy property, and the third follows from concavity of
$\log\det$. The right-hand side of
\eqref{eq:app_block_gaussian_bound} is
$r_{\rm G}(\bfz,\bfQ)$ and is attained by the conditionally i.i.d. Gaussian
implementation with covariance $\bfQ$.

Consequently,
\begin{align}
    R_{{\rm com},L}^{A}
    &\leq
    \int
        r_{\rm G}(\bfz,\bfQ)
        \,\mu_L^{A}(d\bfz,d\bfQ),
    \label{eq:app_rate_A_bound}\\
    R_{{\rm com},L}^{G}
    &=
    \int
        r_{\rm G}(\bfz,\bfQ)
        \,\mu_L^{G}(d\bfz,d\bfQ).
    \label{eq:app_rate_G_equality}
\end{align}
By Lemma~\ref{lem:app_common_invariant_large_L},
Assumption~\ref{ass:mdp_regular}, and continuity of $r_{\rm G}$,
\begin{equation}
    \int
        r_{\rm G}(\bfz,\bfQ)
        \,\mu_L^{G}(d\bfz,d\bfQ)
    -
    \int
        r_{\rm G}(\bfz,\bfQ)
        \,\mu_L^{A}(d\bfz,d\bfQ)
    \longrightarrow0.
\end{equation}
Combining this relation with
\eqref{eq:app_rate_A_bound}--\eqref{eq:app_rate_G_equality} yields
\begin{equation}
    \liminf_{L\to\infty}
    \left(
        R_{{\rm com},L}^{G}
        -
        R_{{\rm com},L}^{A}
    \right)
    \geq0,
\end{equation}
which proves~\eqref{eq:large_block_rate_dominance}.
\end{IEEEproof}

\begin{IEEEproof}[Proof of Proposition~\ref{prop:gaussian_outer_bound_conf}]
For each block length $L$, let $\cK_L^{A}$ denote the set of invariant
state--covariance occupation measures generated by admissible stationary
block-signaling policies, and let
$\cK_L^{G}\subseteq\cK_L^{A}$ denote the subset generated by conditionally
i.i.d. Gaussian signaling.

For an invariant occupation measure $\mu$, define
\begin{equation}
    \bfM(\mu)
    \triangleq
    \int
        \bfJ_D(\Pi,\bfQ)
        \,\mu(d\bfz,d\bfQ)
\end{equation}
and
\begin{align}
    \mathcal J_\lambda(\mu)
    &\triangleq
    \int
        r_{\rm G}(\bfz,\bfQ)
        \,\mu(d\bfz,d\bfQ)
    -
    \lambda
    d\!\left(
        \bfJ(\bfM(\mu))
    \right).
\end{align}
Then
\begin{equation}
    V_{\lambda,L}^{A}
    =
    \sup_{\mu\in\cK_L^{A}}
    \mathcal J_\lambda(\mu),
    \qquad
    V_{\lambda,L}^{G}
    =
    \sup_{\mu\in\cK_L^{G}}
    \mathcal J_\lambda(\mu),
\end{equation}
and
$V_{\lambda,L}^{G}\leq V_{\lambda,L}^{A}$.

Choose $L_n\uparrow\infty$ such that
$V_{\lambda,L_n}^{A}-V_{\lambda,L_n}^{G}$ converges to the superior limit of
$V_{\lambda,L}^{A}-V_{\lambda,L}^{G}$. For each $n$, select
$\mu_n^{A}\in\cK_{L_n}^{A}$ satisfying
\begin{equation}
    \mathcal J_\lambda(\mu_n^{A})
    \geq
    V_{\lambda,L_n}^{A}
    -
    \frac{1}{n}.
\end{equation}
By uniform tightness in
Assumption~\ref{ass:large_block_stationary}, after passing to a subsequence,
$\mu_n^{A}\Rightarrow\mu_\infty$. The invariance argument of
Lemma~\ref{lem:app_common_invariant_large_L} shows that $\mu_\infty$ is
invariant for the common limiting covariance-controlled kernel. Disintegrate
it as
\begin{equation}
    \mu_\infty(d\bfz,d\bfQ)
    =
    \nu_\infty(d\bfz)
    \kappa_\infty(d\bfQ|\bfz),
\end{equation}
where $\kappa_\infty$ is a stationary randomized covariance policy.

Under Assumptions~\ref{ass:mdp_regular} and
\ref{ass:large_block_stationary}, weak convergence preserves the averages of
$r_{\rm G}$ and $\bfJ_D$, while
\eqref{eq:app_fixed_point_continuity} establishes continuity of
$\bfM\mapsto\bfJ(\bfM)$ and $d$ is continuous on the uniformly
positive-definite set of Lemma~\ref{lem:app_J_compact}. Consequently,
\begin{equation}
    \mathcal J_\lambda(\mu_n^{A})
    \longrightarrow
    \mathcal J_\lambda(\mu_\infty).
\end{equation}

Now implement the limiting covariance policy $\kappa_\infty$ by conditionally
i.i.d. Gaussian signaling at each block length $L$, and let
$\mu_{\infty,L}^{G}\in\cK_L^{G}$ be an invariant occupation measure of the
resulting process. By
Lemma~\ref{lem:app_common_invariant_large_L} and the uniqueness condition in
Assumption~\ref{ass:large_block_stationary},
\begin{equation}
    \mu_{\infty,L}^{G}
    \Rightarrow
    \mu_\infty,
    \qquad
    L\to\infty.
\end{equation}
Hence,
\begin{equation}
    \mathcal J_\lambda
    \left(
        \mu_{\infty,L}^{G}
    \right)
    \longrightarrow
    \mathcal J_\lambda(\mu_\infty).
\end{equation}

By near-optimality of $\mu_n^{A}$ and the definition of
$V_{\lambda,L_n}^{G}$,
\begin{align}
    V_{\lambda,L_n}^{A}
    -
    V_{\lambda,L_n}^{G}
    &\leq
    \mathcal J_\lambda(\mu_n^{A})
    -
    \mathcal J_\lambda
    \left(
        \mu_{\infty,L_n}^{G}
    \right)
    +
    \frac{1}{n}.
\end{align}
The right-hand side converges to zero. Since the selected sequence realizes
the superior limit,
\eqref{eq:optimized_gaussian_limsup} follows.

Because
$V_{\lambda,L}^{G}\leq V_{\lambda,L}^{A}$
for every $L$,~\eqref{eq:optimized_gaussian_equality} follows as well.
Finally, applying the argument of
\eqref{eq:finite_block_outer_bound} to the same admissible large-block
signaling class and using
Proposition~\ref{prop:stationary_suffices_conf} gives
\begin{equation}
    J_{\lambda,L}^{\star,\rm LB}
    \leq
    V_{\lambda,L}^{A}
    +
    \epsilon_{Q,L}^{\rm LB}.
\end{equation}
Combining this inequality with
\eqref{eq:optimized_gaussian_equality} and
Assumption~\ref{ass:vanishing_covariance_information} yields
\eqref{eq:gaussian_outer_bound_conf}, completing the proof.
\end{IEEEproof}

\section{EKF Details}
\label{app:EKF}

This appendix summarizes the EKF used for the achievable frontier. The state is
$\bfs_k=[\theta_k,\dot\theta_k,d_k,\dot d_k]^T$, and the dynamics and sensing
model are those in Sections~\ref{sec:system_model} and~\ref{sec:simulations}.

Given $(\hat{\bfs}_{k-1|k-1},\bfP_{k-1|k-1})$, the EKF prediction is
\begin{align}
    \hat{\bfs}_{k|k-1} &= \bfF\hat{\bfs}_{k-1|k-1},\\
    \bfP_{k|k-1} &= \bfF\bfP_{k-1|k-1}\bfF^{T}+\bfQ_w .
\end{align}

For the block observation in~\eqref{eq:block_sensing_numerical}, let
$\bfB_k=\bfQ_k^{1/2}$ denote the Hermitian positive semidefinite square root
of $\bfQ_k$, and define
\begin{equation}
    \widetilde{\bfh}_k(\bfs)
    \triangleq
    \operatorname{vec}\!\left(\bfH_s(\bfs)\bfB_k\right),
\end{equation}
where $\operatorname{vec}(\cdot)$ stacks the columns of its matrix argument.
The derivatives with respect to $\theta_k$ and $d_k$ are
\begin{align}
    \frac{\partial \widetilde{\bfh}_k}{\partial \theta_k}
    &=
    \operatorname{vec}\!\left(
        \alpha(d_k)\bfM_\theta(\theta_k)\bfB_k
    \right), \\
    \frac{\partial \widetilde{\bfh}_k}{\partial d_k}
    &=
    \operatorname{vec}\!\left(
        \alpha(d_k)\bfM_d(\theta_k,d_k)\bfB_k
    \right),
\end{align}
where
\begin{align}
    \bfM_\theta(\theta)
    &\triangleq
    \dot{\bfa}(\theta)\bfa^{\herm}(\theta)
    +
    \bfa(\theta)\dot{\bfa}^{\herm}(\theta), \\
    [\bfM_\theta(\theta)]_{kl}
    &=
    j2\pi\frac{\Delta}{\lambda_c}(k-l)\cos(\theta)
    e^{j2\pi\frac{\Delta}{\lambda_c}(k-l)\sin(\theta)}, \\
    \bfM_d(\theta,d)
    &\triangleq
    \left(-\frac{2}{d}-j\frac{2\pi}{\lambda_c}\right)
    \bfa(\theta)\bfa^{\herm}(\theta).
\end{align}
The derivatives with respect to the velocity components are zero.

Since the observation is complex and the state is real, the EKF uses the
equivalent real-valued model. Let $\widetilde{\bfH}_{J,k}$ denote the complex
Jacobian of $\widetilde{\bfh}_k(\bfs)$ with respect to $\bfs$, evaluated at
$\hat{\bfs}_{k|k-1}$, with zero columns for the velocity components. Define
\begin{equation}
    \bfH_R
    =
    \begin{bmatrix}
        \Re\{\widetilde{\bfH}_{J,k}\}\\
        \Im\{\widetilde{\bfH}_{J,k}\}
    \end{bmatrix},
    \qquad
    \bfR
    =
    \frac{\sigma_{s,\mathrm{blk}}^2}{2}
    \bfI_{2N_{r,s}N_t}.
\end{equation}
Similarly,
\begin{equation}
    \bfy_R
    =
    \begin{bmatrix}
        \Re\{\operatorname{vec}(\widetilde{\bfY}_{s,k})\}\\
        \Im\{\operatorname{vec}(\widetilde{\bfY}_{s,k})\}
    \end{bmatrix},
    \qquad
    \bfh_R
    =
    \begin{bmatrix}
        \Re\{\widetilde{\bfh}_k(\hat{\bfs}_{k|k-1})\}\\
        \Im\{\widetilde{\bfh}_k(\hat{\bfs}_{k|k-1})\}
    \end{bmatrix}.
\end{equation}
The Kalman gain and state update are
\begin{align}
    \bfS_k &= \bfH_R\bfP_{k|k-1}\bfH_R^{T}+\bfR,\\
    \bfK_k &= \bfP_{k|k-1}\bfH_R^{T}\bfS_k^{-1},\\
    \hat{\bfs}_{k|k} &= \hat{\bfs}_{k|k-1}+\bfK_k(\bfy_R-\bfh_R).
\end{align}
The covariance is updated using the Joseph form
\begin{equation}
    \bfP_{k|k}
    =
    (\bfI-\bfK_k\bfH_R)\bfP_{k|k-1}(\bfI-\bfK_k\bfH_R)^{T}
    +
    \bfK_k\bfR\bfK_k^{T}.
\end{equation}
The achievable distortion reported in the simulations is not the EKF covariance;
it is the true Monte Carlo angle and range tracking error in
\eqref{eq:achievable_distortion_sim}.

\bibliographystyle{IEEEtran}
\bibliography{biblio}

\end{document}